%% file: main.tex
\documentclass[sigconf,nonacm]{acmart}

\usepackage[normalem]{ulem}
\usepackage{threeparttable}
\usepackage{graphicx}
\usepackage{multirow}
\usepackage{listings}
\usepackage{amsmath}
\usepackage{amsthm}
\usepackage{subcaption}
\usepackage[ruled,vlined]{algorithm2e}
\usepackage{xcolor}

\definecolor{codegreen}{rgb}{0,0.6,0}
\definecolor{codegray}{rgb}{0.5,0.5,0.5}
\definecolor{codepurple}{rgb}{0.58,0,0.82}
\definecolor{backcolour}{rgb}{0.95,0.95,0.92}

\definecolor{carrotorange}{rgb}{0.93, 0.57, 0.13}
\definecolor{burntorange}{rgb}{0.8, 0.33, 0.0}

\lstdefinestyle{mystyle}{
    backgroundcolor=\color{backcolour},   
    commentstyle=\color{codegreen},
    keywordstyle=\color{magenta},
    numberstyle=\tiny\color{codegray},
    stringstyle=\color{codepurple},
    basicstyle=\ttfamily\footnotesize,
    breakatwhitespace=false,         
    breaklines=true,                 
    captionpos=b,                    
    keepspaces=true,                 
    numbers=left,                    
    numbersep=5pt,                  
    showspaces=false,                
    showstringspaces=false,
    showtabs=false,                  
    tabsize=2
}

\newtheorem{theorem}{Theorem}[section]

\renewcommand*{\proofname}{Proof Sketch}

\setcopyright{none}
\renewcommand\footnotetextcopyrightpermission[1]{} 
\begin{document}

\title{Search for Optimal Systolic Arrays: A Comprehensive Automated Exploration Framework and Lessons Learned}



\author{Jie Wang}
\affiliation{%
  \institution{University of California, Los Angeles}
}
\email{jiewang@cs.ucla.edu}

\author{Jason Cong}
\affiliation{\institution{University of California, Los Angeles}}
\email{cong@cs.ucla.edu}

\begin{abstract}

\input{tex/abstract}

\end{abstract}

\settopmatter{printfolios=true}
\maketitle

\input{tex/intro}
\input{tex/background}
\input{tex/construct}
\input{tex/scan}

\input{tex/results}
\input{tex/discuss}
\input{tex/conclude}


\begin{acks}
\small
We would like to thank Marci Baun for helping edit the paper. This work is partially supported by the Intel and NSF joint research center for Computer Assisted Programming for Heterogeneous Architectures (CAPA), NSF NeuroNex Award DBI-1707408, and the members from the CDSC Industrial Partnership Program. We acknowledge the valuable support of the Xilinx Adaptive Compute Clusters (XACC) Program.
\end{acks}


\bibliographystyle{plain}
\bibliography{references}

\end{document}

%% file: tex/abstract.tex
Systolic arrays have been widely used for accelerating HPC and deep learning applications. 
There is a plethora of previous works on the performance tuning of systolic arrays, but usually based on a number of oversimplified assumptions (e.g., only considering divisors for loop tiling, pruning based on off-chip data communication) to reduce the design space. 

In this paper, we present a comprehensive design space exploration tool named Odyssey for systolic array optimization.  
Odyssey does not rely on artificial assumptions to limit the design space, and yet it is highly efficient and scalable with a hybrid optimization technique.  For example, for a 1024×1024×1024 matrix multiplication, it finds designs that reach 90\% of the optimal performance in 5 seconds with a single CPU thread. 
Moreover, using Odyssey, we unveil and quantify the suboptimality introduced by multiple commonly used oversimplifications in prior studies for systolic array design space exploration.   
For example, Odyssey results show that limiting to divisors for loop tiling leads to a 39\% performance loss, and pruning based on off-chip data movement results in a 45\% performance loss. We applied Odyssey to explore the architecture trade-offs for matrix multiplication and convolutional neural network, providing inspiration into possible optimizations for these two applications.


%% file: tex/intro.tex
\section{Introduction}
\label{sec:intro}
Performance optimization for a given class of microarchitectures, also called performance tuning, has long been an important topic given the complexity of hardware systems and applications. 
This issue is intensified on domain-specific architectures (DSA), which grant designers explicit control over the software stack and hardware architecture, opening up a vast design space to explore. 

This paper focuses on the performance tuning of systolic arrays. 
A complete design space of systolic arrays contains multiple dimensions, such as the selection of \textit{dataflows}\footnote{Given a target loop program to map to systolic arrays, designers could select different loops to be mapped to spatial dimensions of systolic arrays, leading to different array topologies and execution models. We define each of such choices as a unique dataflow.}, loop permutation and tiling.
These factors impact the final design performance in an intertwined manner and compose a vast design space which is intractable to handle manually.


Many previous works have attempted this challenging task by looking into different dimensions of the design space and proposing various auto-tuning methods~\cite{timeloop,confuciux,cosa,dmazerunner,marvel,tenet,interstellar,gamma}. However, after a thorough examination of the previous works, we identified several limitations that need to be addressed.


\begin{figure}[t]
    \centering
    \begin{subfigure}{0.45\columnwidth}
        \centering
        \includegraphics[width=\linewidth]{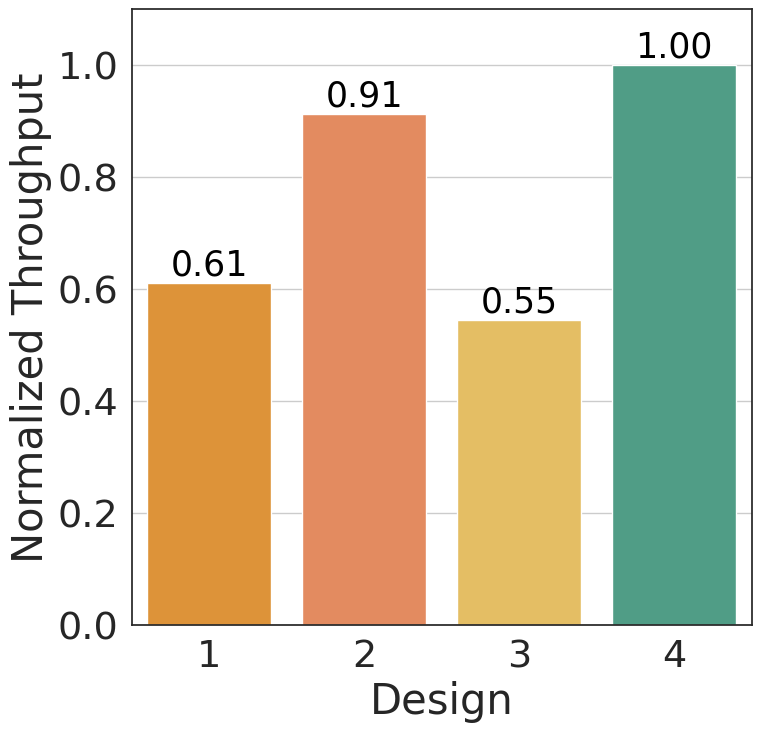}
        \caption{Throughput}
    \end{subfigure}
    \begin{subfigure}{0.45\columnwidth}
        \includegraphics[width=\linewidth]{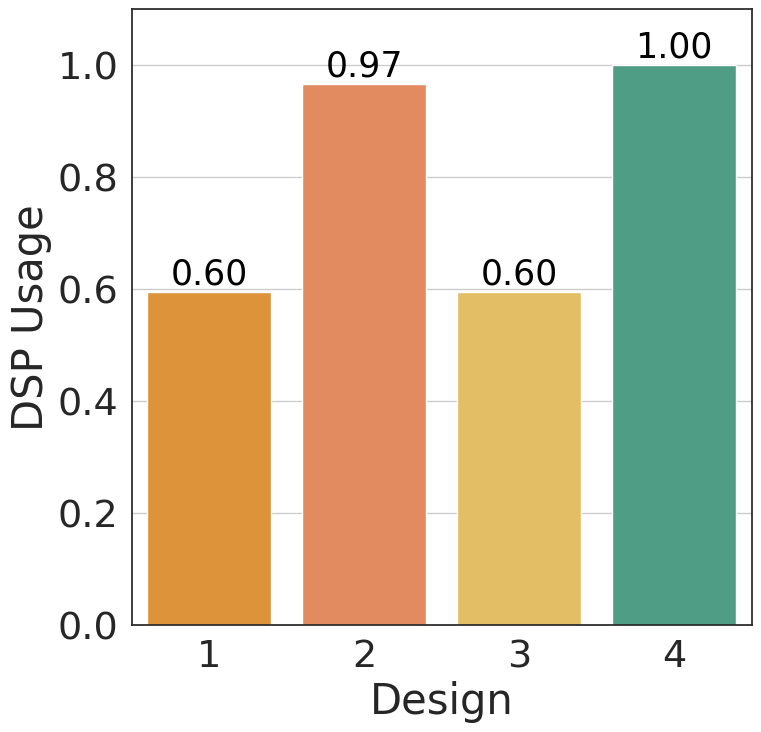}
        \caption{DSP}
        \label{fig:motive_dsp}
    \end{subfigure}
    \vspace{-0.05in}
    \caption{Throughput and DSP usage of different systolic array designs for a 1024x1024x1024 MM. Design 4 is produced by Odyssey.}
    \vspace{-0.05in}
    \label{fig:motiv_throughput}
\end{figure}

\textbf{Limitation 1: Incomplete coverage of the design space.} 
When selecting the tiling factors, many previous works only considered divisors to reduce the design space~\cite{dmazerunner,interstellar,marvel,cosa}. 
In Figure~\ref{fig:motiv_throughput}, we compare the throughput and DSP usage of best systolic arrays found when limiting tiling factors to 1) divisors only (Design 1) and 2) both divisors and non-divisors (Design 4) for a $1024\times 1024\times 1024$ matrix multiplication (MM)\footnote{Please refer to Section~\ref{sec:evo_details} for more details.}. Limiting to divisors leads to a 39\% performance loss. With the limited design space, the divisor-only design fails to fully exploit the on-chip resource, with only 60\% DSP usage.

\textbf{Limitation 2: Inaccurate performance modeling.} An inaccurate performance model could also hurt the quality of search results. For example, the previous work TENET~\cite{tenet} estimated the design latency as the maximum of compute and communication latency. This model overlooks the prologue/epilogue phases when loading the first tile of data and writing out the final results. Figure~\ref{fig:motiv_throughput} shows the performance of the best design found when using the maximum-based model (Design 2) for the same MM problem. We observe a 9\% performance loss compared to the optimal design.

\textbf{Limitation 3: Inefficient search methods and imperfect pruning heuristics.} 
When searching the design space, many previous works adopted pruning-based exhaustive search which may not scale to large-sized problems~\cite{dmazerunner,interstellar,marvel,tenet}. To make matters worse, several works chose imperfect pruning heuristics, failing to cover optimal designs. For example, the previous work Marvel~\cite{marvel} pruned the design space based on the off-chip data communication and applied an exhaustive search in the pruned sub space. However, an optimal design needs to balance both the data communication and computation and does not necessarily minimize the off-chip memory accesses. As a result, the best design found which minimizes the off-chip data communication (Design 3) in Figure~\ref{fig:motiv_throughput} introduces a 45\% performance loss compared to the optimal design.


All of these limitations affect the quality of search results and further impact the architecture decisions that designers derive based on these results. 
To overcome this challenge, in this paper, we propose a new automatic design space exploration framework for systolic arrays, \textit{Odyssey}
\footnote{Odyssey is abbreviated from \underline{AU}tomatic \underline{DE}sign space exploration for \underline{SY}stolic arrays.}. 
Figure~\ref{fig:tuning_flow} depicts the proposed tuning flow.

\begin{figure}[t]
\includegraphics[width=0.75\columnwidth]{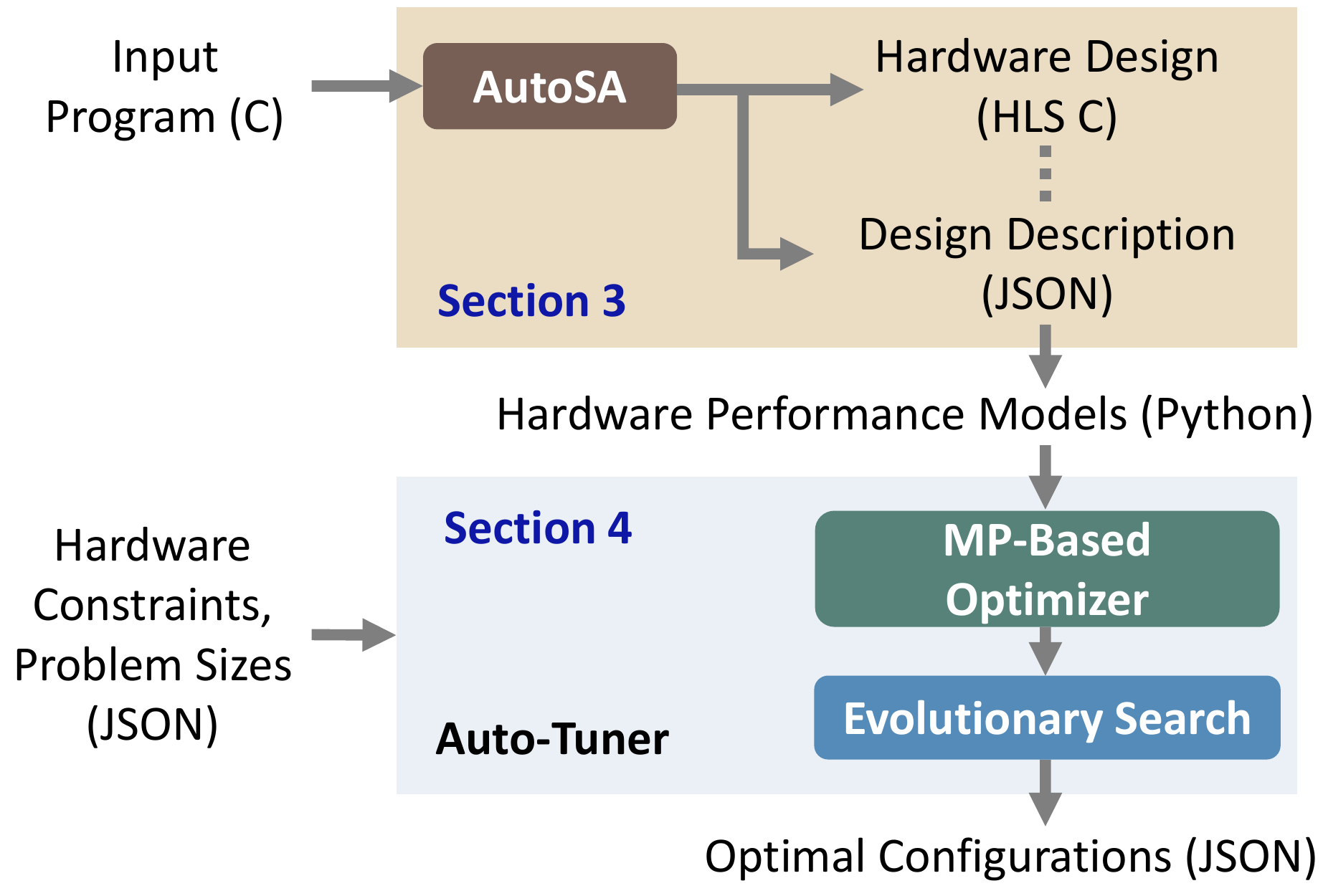}
\centering
\vspace{-0.05in}
\caption{Overview of Odyssey framework.}
\vspace{-0.05in}
\label{fig:tuning_flow}
\end{figure}

Odyssey leverages AutoSA~\cite{autosa}, an open-source FPGA-based systolic array compiler, to construct the design space automatically. 
AutoSA takes in a C program that describes the target algorithm to map to systolic arrays and generates the systolic array designs in Xilinx HLS C~\cite{xilinx_hls}.
We extend the AutoSA framework to generate a design description file that covers the full details of the generated hardware. 
Odyssey uses this file to create hardware performance models as symbolic expressions of the tuning parameters that can be used by the auto-tuner.
Inside the auto-tuner, Odyssey implements a two-stage flow that starts with a mathematical programming (MP)-based optimizer that leverages optimization solvers with a simplified objective function to produce an initial high-quality design, followed by the evolutionary search with the accurate performance models. Odyssey surpasses the previous works in multiple dimensions. 

\textbf{Contribution 1: A comprehensive design space and accurate performance modeling.} 
Odyssey covers a comprehensive design space including the selection of dataflows, loop permutation and tiling. Furthermore, Odyssey derives the performance models directly from the compiler-generated hardware with high estimation accuracy.
For example, we are able to produce latency models by traversing the ASTs of the generated HLS designs, which achieved a low estimation error of 1.99\% compared to the real hardware execution.

\textbf{Contribution 2: Efficient auto-tuning methods.} The proposed auto-tuning algorithm does not rely on any artificial assumptions (e.g., divisors tiling factors, pruning based on off-chip data communication) to limit the design space, and yet it is highly efficient with a hybrid optimization technique that combines the MP-based optimization and evolutionary search. 
We propose a hybrid mutation operation in evolutionary search that takes non-divisor tiling factors into consideration. In addition, we implement a pruning algorithm for loop permutation that trims away inferior designs without omitting the optimal designs.
As a result, our search framework can locate designs with high performance in a short amount of time, For example, for the $1024\times1024\times1024$ MM, the proposed auto-tuner finds a design that achieves 90\% of the optimal performance in 5 seconds with a single CPU thread.

\textbf{Contribution 3: A fully automated and open-source framework.} 
The entire flow is fully automated and open-sourced\footnote{\url{https://github.com/UCLA-VAST/AutoSA/tree/master/autosa_scripts/odyssey}}. Odyssey is the first work that is built directly on an open-source systolic array compiler to construct the design space and generate performance models based on real hardware implementations. This guarantees the comprehensiveness and accuracy of the design space modeling, and validity and reproducibility of the search results.
In this paper, we show two architecture studies on MM and convolutionary neural network (CNN). 


The rest of this paper is organized as follows: Section~\ref{sec:background} discusses the previous works and covers the basics of AutoSA.
Section~\ref{sec:construct} explains how we construct the design space. In Section~\ref{sec:scan}, we introduce the search methods to explore this design space. Section~\ref{sec:evaluate_results} presents the evaluation results. Section~\ref{sec:discuss} discusses lessons learned from this work and Section~\ref{sec:limit} states the limitations of this work.
Section~\ref{sec:conclude} concludes the paper.

%% file: tex/background.tex
\section{Background and Related Work}
\label{sec:background}

There is a plethora of previous works on performance tuning of systolic arrays~\cite{timeloop,confuciux,cosa,dmazerunner,marvel,tenet,interstellar,gamma,hegde2021mind}. Table~\ref{table:prior_work_tuning} lists several recent works. Note that some of the listed works covered a broader set of architectures beyond systolic arrays. In this work, we only focus on the systolic array architecture. However, the methodology proposed in this work can be applied to other architectures as well. We discuss the prior works from two dimensions: the design space and search methods.

\begin{table}[t]
\caption{Comparison between different systolic array architecture performance tuning frameworks.}
\resizebox{\columnwidth}{!}{
\begin{tabular}{lcccc}
\toprule
                                 & Design Space             & \multicolumn{2}{c}{Performance Models}                                                     & \multirow{2}{*}{Search Methods}                                                     \\ \cline{1-4}
                                 & Non-Divisors             & \begin{tabular}[c]{@{}c@{}}Prologue/\\ Epilogue\end{tabular} & Generation                  &                                                                                     \\ \toprule
Timeloop~\cite{timeloop}         & \textcolor{red}{N}       & \textcolor{red}{N}                                           & \textcolor{red}{Manual}     & \begin{tabular}[c]{@{}c@{}}Exhaustive w/ Pruning\\ Random Search\end{tabular}       \\ \hline
dMazeRunner~\cite{dmazerunner}   & \textcolor{red}{N}       & \textcolor{red}{N}                                           & \textcolor{red}{Manual}     & Exhaustive w/ Pruning                                                               \\ \hline
Interstellar~\cite{interstellar} & \textcolor{red}{N}       & N/A                                                          & \textcolor{red}{Manual}     & Exhaustive w/ Pruning                                                               \\ \hline
Marvel~\cite{marvel}             & \textcolor{red}{N}       & {\color{green!55!blue}Y}                                     & \textcolor{red}{Manual}     & \begin{tabular}[c]{@{}c@{}}Mathematical Programming\\ Exhaustive w/ Pruning\end{tabular} \\ \hline
ConfuciuX~\cite{confuciux}       & N/A                      & {\color{green!55!blue}Y}                                     & \textcolor{red}{Manual}     & \begin{tabular}[c]{@{}c@{}}RL\\ Evolutionary Search\end{tabular}                    \\ \hline
CoSA~\cite{cosa}                 & \textcolor{red}{N}       & {\color{green!55!blue}Y}                                     & \textcolor{red}{Manual}     & Mathematical Programming                                 \\ \hline
TENET~\cite{tenet}               & N/A                      & \textcolor{red}{N}                                           & \textcolor{red}{Manual}     & Exhaustive w/ Pruning                                                               \\ \hline
\textbf{Odyssey}                          & {\color{green!55!blue}Y} & {\color{green!55!blue}Y}                                     & {\color{green!55!blue}Auto} & \begin{tabular}[c]{@{}c@{}}Mathematical Programming\\ Evolutionary Search\end{tabular}   \\ \bottomrule
\end{tabular}
}
\label{table:prior_work_tuning}
\end{table}

\subsection{Design Space}
\paragraph{Design space coverage} We consider three dimensions of the design space: dataflows, loop permutation and loop tiling. In addition, in this work, we cover non-divisor tiling factors in loop tiling. Many previous works only used divisors for simplicity~\cite{timeloop,dmazerunner,interstellar,marvel,cosa}. However, this choice could lead to significant performance loss, as discussed in Section~\ref{sec:intro}.

\paragraph{Performance model} 
The accuracy of performance models plays an important role in performance tuning. In Section~\ref{sec:intro}, we discussed the issue of using a simplified performance model that overlooks the epilogue and prologue phases of the hardware execution. 
Table~\ref{table:prior_work_tuning} highlights several previous works with a similar issue~\cite{timeloop,dmazerunner,tenet}. In addition, such performance models are usually derived manually which is time-consuming and inaccurate. Odyssey distinguishes itself from the prior works in that it automatically creates performance models by leveraging the AutoSA compiler. 

\subsection{Search Methods}
\label{sec:prev_search_method}
We classify the search methods into three categories.

\paragraph{Brute-force methods} Such methods include random search~\cite{timeloop,autotvm} and exhaustive search with pruning~\cite{xuechao_sa,autosa,dmazerunner,interstellar,marvel,tenet}. Brute-force methods face scalibility issues with large-scale problems.
For pruning-based exhaustive search, imperfect pruning heuristics could lead to inferior designs. 
In Section~\ref{sec:intro}, we mentioned using off-chip data communication as the pruning heuristic leads to inferior performance. Previous works like~\cite{marvel,chen_communication} used this heuristic.

\paragraph{Mathematical programming-based methods} These methods formulate the search task as a mathematical optimization problem and resolve it with optimization solvers~\cite{analytical_cnn,analytical_tc,marvel,cosa}. 
The recent work CoSA~\cite{cosa} modeled the design space as a mixed integer programming (MIP) problem. To accommodate for the formulation requirements of MIP, CoSA set the objective function as a linear combination of several high-order factors such as resource utilization and data communication.
The simplified models employed in such approaches will lead to inferior designs. We conducted an experiment in which we model the tuning task for AutoSA-generated designs as a non-linear optimization problem and used the off-the-shelf solver (AMPL~\cite{ampl} with Ipopt~\cite{ipopt}) to generate the solutions. The best design obtained from this approach is 1.5$\times$ slower than the optimal design obtained by Odyssey (see Section~\ref{sec:solver_details}). 

\paragraph{Iterative methods} 
Examples include simulated annealing~\cite{autotvm}, evolutionary search~\cite{ansor}, Bayesian optimization~\cite{david_bayesian}, and reinforcement learning~\cite{confuciux}. In this work, we examined several iterative search methods and chose evolutionary search given its high sample efficiency and search quality.


An ideal performance tuning framework should achieve: 1) a \textit{comprehensive coverage} and \textit{accurate modeling} of the design space, and 2) \textit{efficient} search methods to explore the design space. The failure in either of the two targets will impact the quality of search results, as well as the architecture conclusions derived from these results. Unfortunately, regardless of the plethora of past studies, we observe no prior work that reached a balance between these two goals. This situation has motivated us to tackle this challenge. 

\subsection{Review of Automatic Systolic Array Generation}
\label{sec:systolic_generate}
Automatic systolic array generation is an important research topic given the high performance of the systolic array architecture and the complexities of the designing process~\cite{gemmini,mmalpha,uday_ppopp,calyx,susy}. The recent work, AutoSA~\cite{autosa}, reported the best performance results in this field.
AutoSA takes in a C program as the input and applies a sequence of program transformations on this program to build and optimize systolic arrays.

\begin{figure*}[t]
\includegraphics[width=\textwidth]{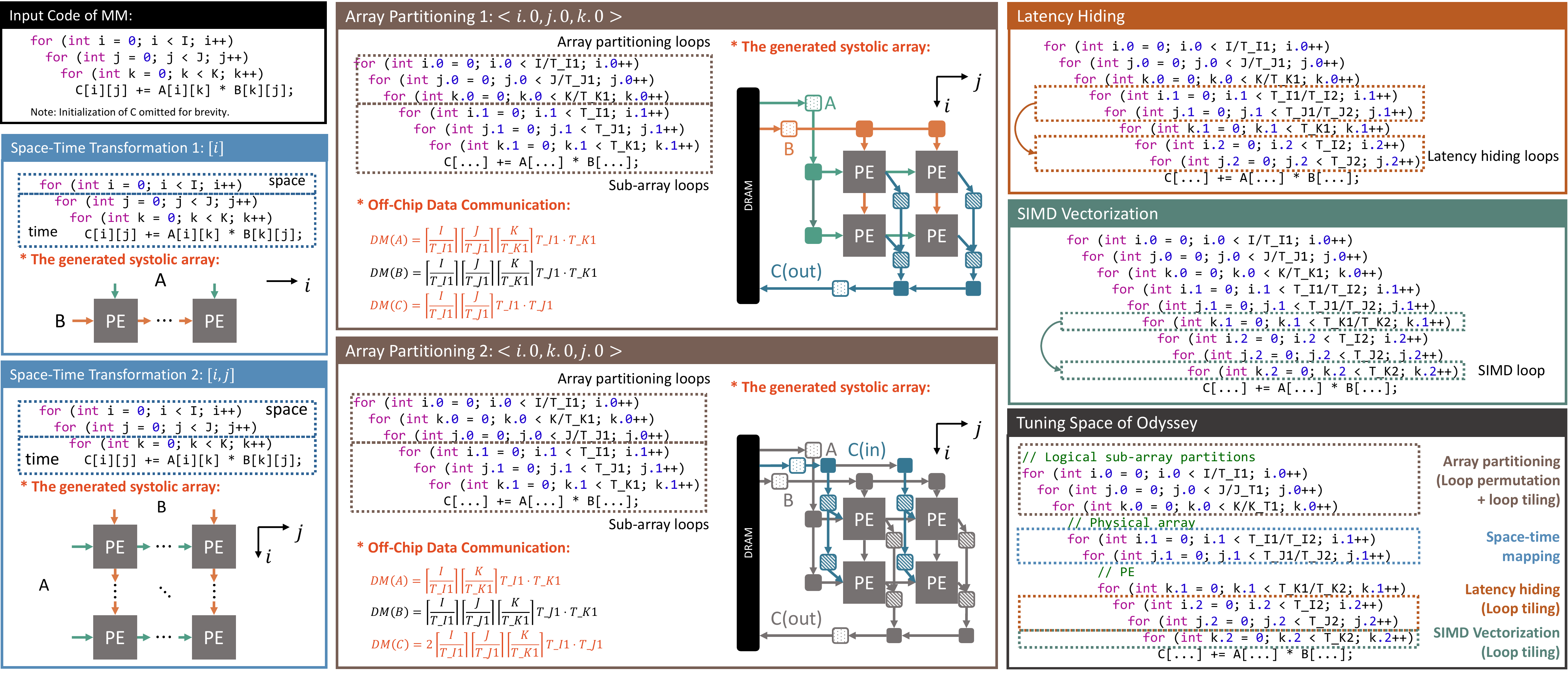}
\centering
\caption{Example of using AutoSA to generate systolic arrays for MM.}
\label{fig:autosa_flow}
\end{figure*}

The first step \textit{space-time mapping} assigns space and time semantics to different loops to create a systolic array. The left column in Figure~\ref{fig:autosa_flow} shows two examples for MM. The first example assigns loop $i$ as space loops, leading to a 1D systolic array with $I$ PEs. 
We assign all the loops below the space loops as time loops that describe the computation inside PEs.
The second example maps loops $i$ and $j$ to spatial dimensions, producing a 2D systolic array with a size of $I\times J$. 
AutoSA identifies all the legal space loops and generates different space-time mapping candidates
\footnote{Previous works~\cite{eyeriss,interstellar,maestro} used the term \textit{dataflow} to identify different array topologies and execution patterns, which are equivalent to different space-time mappings within the scope of systolic arrays. 
In the rest of the paper, we use \textit{dataflow} to refer to different space-time mappings. We annotate each dataflow in the format of $[i,j]$ that marks the selected space loops for this array.}.

The next step, \textit{array partitioning}, employs loop tiling on the outermost permutable loops to reduce the array size.
The middle column in Figure~\ref{fig:autosa_flow} presents two examples of MM in which the original loops $i$, $j$, $k$ are tiled with factors of $[2,2,2]$ that leads to a smaller array with a size of $2\times 2$. The original computation task is partitioned and executed on this array in sequence. The tiled loops, named as \textit{array partitioning loops}, can be further permuted which could lead to different array architectures. The first array keeps the original loop ordering $<i.0,j.0,k.0>$. Data of matrix $C$ are accumulated along the loop $k.0$. Consequently, intermediate results are updated on-chip and only written out off-chip on the last iteration of loop $k.0$. In comparison, the second array uses the loop ordering $<i.0,k.0,j.0>$. As different tiles of matrix $C$ are updated at each iteration of loop $j.0$, we add additional hardware modules in the on-chip I/O network (marked as $C(in)$) to load in the intermediate results. 
The third step, \textit{latency hiding}, tiles parallel loops and permutes them to the innermost to hide the computation latency. And the last step, \textit{SIMD vectorization}, vectorizes one loop to improve the compute-over-control ratio to amortize the control overheads inside PEs.



With a comprehensive coverage of hardware optimization techniques, AutoSA generates high-performance systolic arrays with comparable or better performance than previous works~\cite{autosa}. However, such a vast design also poses significant challenges to performance tuning. As an example, considering all the available tuning options, the size of design space bloats to $O(2^{40})$ for a $1024\times 1024\times 1024$ MM. 
This challenge has motivated us to develop Odyssey which provides efficient auto-tuning support to explore such a design space.

The subfigure at the bottom-right of Figure~\ref{fig:autosa_flow} summarizes the design space covered by Odyssey. Odyssey explores different dataflows in the space-time transformation and the loop tiling factors in array partitioning, latency hiding, and SIMD vectorization. Odyssey also explores non-divisor tiling factors and loop permutation in the step of array partitioning. When a non-divisor tiling factor is chosen, the original problem dimension is zero padded to be a multiple of the tiling factor. 
Tiling factors of latency hiding and SIMD vectorization are still required to be divisors as the array dimension is fixed after array partitioning and non-divisor tiling factors at later steps introduce prologue/epilogue phases and costly max/min operations that hurt the design performance. 


%% file: tex/construct.tex
\section{Constructing the Design Space}
\label{sec:construct}
In this section, we discuss how we construct the design space. 
The current design space considers loop permutation. 
We can prune the loop permutation orderings without filtering out the optimal designs.
The second part of this section covers the details of the pruning heuristics.

\subsection{Enhancement to AutoSA}
We have extended AutoSA to generate a design descriptor that contains all the necessary information for estimating the design performance.
This descriptor contains the following components:
\begin{itemize}
    \item {\textit{Abstract syntax tree (AST)}}: ASTs of all the hardware modules for latency estimation.
    \item {\textit{Memory information}}: Properties of on-chip buffers including buffer size, data type, and memory banking.
    \item {\textit{Compute information}}: Information about the PE computation logic, such as the SIMD lanes and computation statements. 
    \item {\textit{Array Topology}}: Information about the topology of the systolic array, such as the array dimension and the number of different hardware modules. We use this information, along with the memory and compute information for resource estimation.
    \item {\textit{Tuning parameters}}: Loop tiling factors to be tuned.
\end{itemize}

The auto-tuner utilizes this description file to create a Python file containing functions for estimating the design performance. All the performance models are symbolic expressions of the tuning parameters.
During the search, the auto-tuner samples the design space and plugs in different tuning parameters into the performance models for assessing the design performance. 


%

\subsection{Loop Permutation Pruning}
\label{sec:loop_permute}
Odyssey explores different loop permutation orderings in the array partitioning. As shown in Section~\ref{sec:systolic_generate}, different loop orderings may lead to various array architectures. AutoSA enumerates all the loop permutation orderings. For MM, there are $3!=6$ different orderings to consider. The number grows larger for a more complicated application like CNN. 

\begin{lstlisting}[language=c, caption=Example code of a CNN layer. Batch size and stride are set to 1 here for simplicity.,label=lst:cnncode,basicstyle=\tiny]
for (int o = 0; o < O; o++)
  for (int h = 0; o < H; h++)
    for (int w = 0; o < W; w++) {
      fo[o][h][w] = 0;    
      for (int i = 0; o < I; i++)
        for (int p = 0; o < P; p++)
          for (int q = 0; o < Q; q++)
            fo[o][h][w] += fi[i][h+p][w+q] * w[o][i][p][q];
      }
\end{lstlisting}

Listing~\ref{lst:cnncode} shows the example code of one CNN layer. The six-level nested loop leads to $6!=720$ different loop orderings. However, as pointed out by previous works~\cite{analytical_cnn,analytical_tc,dmazerunner}, among all the loop orderings, many of them are dominated by a few orderings in performance, thus can be safely pruned without leaving out the optimal points. Next, we show that with proper pruning, we can reduce the number of loop orderings to consider for both MM and CNN to only 3. 


We consider resource usage and latency when assessing the design performance. Different loop orderings will impact the structure of the I/O network, resulting in different resource usage. In Figure~\ref{fig:autosa_flow}, we showed that after hoisting the loop $k.0$ from the innermost position of the array partitioning loop band, additional I/O modules for transferring the intermediate results of matrix $C$ are added, increasing the total resource usage. 
The key takeaway is: \textit{By placing loops that carry the flow dependences innermost in the array partitioning loop band, intermediate data are accumulated on-chip, eliminating the resource overheads brought by the additional I/O modules.} 



Latency-wise, different loop orderings will affect off-chip data communication. We compute the off-chip data movement for the two example designs in Figure~\ref{fig:autosa_flow}. For the first ordering $<i.0,j.0,k.0>$, final results of matrix $C$ are only drained out at the last iteration of loop $k.0$, leading to a total amount of data movement as:
\begin{equation}
\small
\vspace{-0.025in}
DM(C)=\lceil \frac{I}{T\_{I1}} \rceil \lceil \frac{J}{T\_{J1}} \rceil T\_{I1} \cdot T\_{J1}
\vspace{-0.025in}
\end{equation}

As for the second ordering $<i.0,k.0,j.0>$, the intermediate results of matrix $C$ are swapped off-chip at each loop iteration of $j.0$. We compute the total data movement of matrix $C$ as: 
\begin{equation}
\small
\vspace{-0.025in}
DM(C)=2 \lceil \frac{I}{T\_{I1}} \rceil \lceil \frac{K}{T\_{K1}} \rceil \lceil \frac{J}{T\_{J1}} \rceil T\_{I1} \cdot T\_{J1}
\vspace{-0.025in}
\end{equation}

To take into account both the inbound and outbound traffic of matrix $C$, we multiply the factor 2.
Compared to the first ordering, the second ordering introduces a higher amount of data movement for matrix $C$. For communication-bound designs, this could lead to a longer latency.
In addition to the loops that carry the flow dependence, loops that carry the read dependence will impact the data communication as well. We use matrix $A$ as an example. As shown in Figure~\ref{fig:autosa_flow}, with the loop ordering $<i.0, j.0, k.0>$, at each array partition, we load new array tiles with a size of $T\_{I1}\times T\_{K1}$ from array $A$. In comparison, with the loop ordering $<i.0,k.0,j.0>$, data of matrix $A$ are reused along the loop $j.0$. New data tiles are only loaded at each new loop iteration of loop $k.0$, reducing the data communication for matrix $A$ compared to the first ordering. Detailed equations of data movement for matrix $A$ can be found in Figure~\ref{fig:autosa_flow}. 
The key takeaway is: \textit{By placing array partitioning loops that carry the flow/read dependences innermost, data are reused on-chip, reducing the off-chip data communication and design latency.} 






Putting it all together, we have a complete picture of the pruning strategy.
\begin{theorem}
Given a program that can be mapped to systolic arrays, let $RL(r)$ be the set of array partitioning loops that carry the read/flow dependences associated with the array reference $r$, and $NRL(r)$ be the rest of the loops in the array partitioning loop band, the set of unique loop orderings $O$ can be obtained as the union of loop orderings in the form of $<NRL(r),RL(r)>$ for each array reference $r$. All the other loop orderings are dominated by $O$ in terms of resource usage and latency. Note that $RL(r)$ could be an empty set if there is no read/flow dependence associated with the $r$. For this case, the loop ordering is in the form of $<NRL(r)>$, and is added into $O$ for consideration.
\end{theorem}

\begin{proof}
Assume the above statement is false, i.e., there is a loop ordering $o'$ out of the set $O$ that achieves better performance than loop orderings in $O$. Then, there exists at least one loop $l_{rl}$ in $o'$ that carries the read/flow dependences for a certain array reference $r$, and is placed above a certain loop $l_{nrl}$ that belongs to $NRL(r)$. We group all such loops $l_{rl}$ into a set $RL'(r)$ and permute them to the innermost of the array partitioning loop band to generate a new loop ordering $\hat{o}$ that belongs to $O$. If $r$ is associated with flow dependences, $o'$ introduces additional I/O modules for loading in the intermediate results, increasing the resource usage compared to $\hat{o}$. If $r$ is associated with read/flow dependences, $o'$ increases off-chip data communication, and could lead to a longer latency than $\hat{o}$ if the design is bound with data communication of $r$. Overall, this loop ordering $o'$ is dominated by $\hat{o}$ in both resource usage and design latency, which contradicts to the initial assumption that $o'$ dominates loop orderings from $O$ in performance.
\end{proof}

For the MM example, loop $i.0$ carries the read dependence for array $B$, loop $j.0$ carries the read dependence for array $A$, and loop $k.0$ carries the flow dependence for array $C$. In total, there are three loop ordering candidates which lead to systolic arrays with potentially different performance. Note that as long as the innermost loop is fixed, the ordering of other loops will not impact the performance. In the rest of this paper, we use the annotation $<[i.0,j,0], [k.0]>$ to identify the set of loop orderings. All the loops in the same brackets can be permuted freely with equivalent performance. We will choose one ordering randomly in practice.
Table~\ref{table:tuning_designs} shows all the unique loop orderings for MM and CNN.

\begin{table}[t]
\centering
\caption{Unique systolic arrays generated by AutoSA for MM and CNN.}
\vspace{-0.05in}
\resizebox{\columnwidth}{!}{
\begin{tabular}{ccc}
\toprule
Application      & MM                                                                                                                & CNN                                                                                                                                             \\ \toprule
Dataflows        & {[}i{]}, {[}j{]}, {[}k{]}, {[}i,j{]}, {[}i,k{]}, {[}j,k{]}                                                        & \begin{tabular}[c]{@{}c@{}}{[}o{]}, {[}h{]}, {[}w{]}, {[}i{]}, \\ {[}o,h{]}, {[}o,w{]}, {[}o,i{]}, {[}h,w{]}, {[}h,i{]}, {[}w,i{]}\end{tabular} \\ \midrule
Loop Permutation & \textless{}{[}i,j{]},k\textgreater{}, \textless{}{[}j,k{]},i\textgreater{}, \textless{}{[}i,k{]}, j\textgreater{} & \textless{}{[}o,h,w{]},{[}i,p,q{]}\textgreater{}, \textless{}{[}o,i,p,q{]},{[}h,w{]}\textgreater{}, \textless{}{[}i,h,w,p,q{]},o\textgreater{}  \\ \midrule
\#Designs        & 18                                                                                                                & 30                                                                                                                                              \\ \bottomrule
\end{tabular}
}
\label{table:tuning_designs}
\end{table}

To summarize, given an input program, we use AutoSA to generate different dataflows and loop permutation orderings of the array partitioning loops, and leave the tiling factors as tunable parameters to be handled by the auto-tuner. 
The next section will discuss the details of the auto-tuner.

%% file: tex/scan.tex
\section{Searching the Design Space}
\label{sec:scan}

\subsection{Evolutionary Search}
\label{sec:evo_details}
In the Odyssey framework, we select evolutionary search as the backbone search method. 
Evolutionary search~\cite{genetic} is a generic meta-heuristic algorithm inspired by biological evolution, in which individuals of a population gradually improve themselves through a series of biological mechanisms such as \textit{mutation}, \textit{crossover}, and \textit{selection}. 
In the context of hardware design space exploration, each design point is termed as an individual in the population. The feature vector that describes a design point is coined as the genome of the individual. Each design point is assigned a fitness score by its performance assessed by the real hardware measurement or the performance model. At each iteration, we select a group of individuals with the highest fitness scores in the population as the parents to produce the children in the next iteration (\textit{selection}). These parents will breed new individuals with \textit{mutation} and \textit{crossover} operations. In the next iteration, we filter out the low-fitted individuals and repeat the same reproduction process until finding the satisfactory solutions. 
\paragraph{Encoding scheme}
Each individual is encoded by the tiling factors used in AutoSA compilation passes. 

\paragraph{Mutation} Figure~\ref{fig:evo_mutation} shows one example for MM.
We first compute the updated loop bounds with the current tiling factors. The mutation process operates on the values of the loop bounds.
We implement two mutation methods: \textit{factorization-based} mutation and \textit{random} mutation. 

\begin{figure}[t]
\includegraphics[width=\columnwidth]{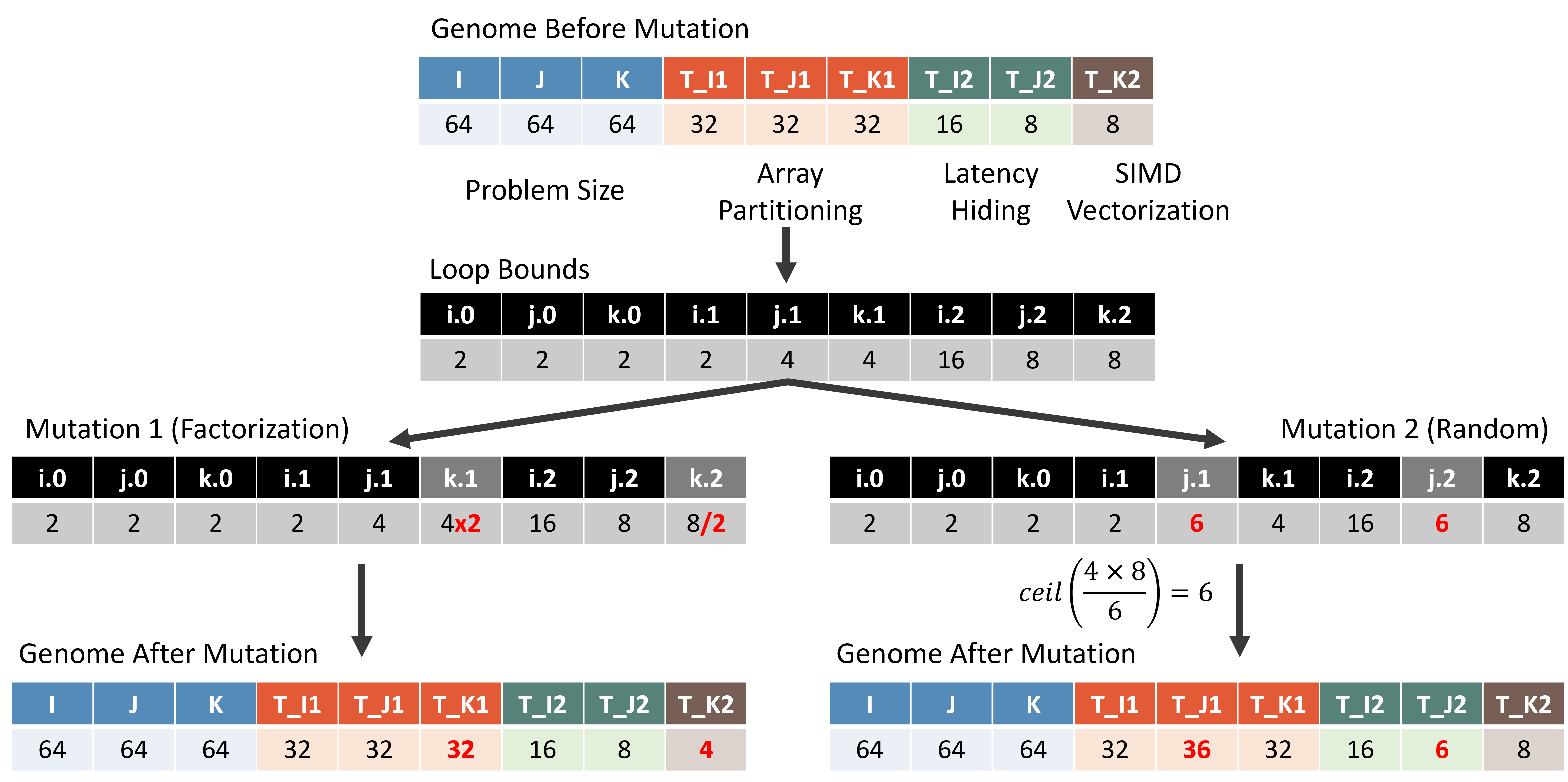}
\centering
\caption{Example of different mutation methods.}
\label{fig:evo_mutation}
\vspace{-0.05in}
\end{figure}

\textit{1) Factorization-based mutation:} We randomly choose one loop $l_1$, divide it by a random divisor $\alpha$ ($\alpha|l_1$), and multiply $\alpha$ to another loop $l_2$ that is derived from the same loop as $l_1$. For example, in Figure~\ref{fig:evo_mutation}, we select the loop $k.2$ and divide it by 2, then we choose the other loop $k.1$ and multiply its current value with 2. The factorization-based mutation keeps the product of the tiled loop sizes unchanged, guaranteeing the muted program is always valid.


The factorization-based mutation always chooses the divisors of the loop bounds. While such an implementation is common in many previous works~\cite{cosa,timeloop,dmazerunner,interstellar}, it could result in a reduced design space with inferior performance. 
Table~\ref{table:tuning_mutation_cmp} compares the best systolic arrays identified by two search methods for a $1024\times1024\times1024$ MM. The first method uses the factorization-based mutation that only considers divisor tiling factors. The second method employs the hybrid mutation method which we will introduce soon that considers non-divisor tiling factors. As shown in the table, the first design found with the factorization-based mutation has a 39\% performance gap compared to the second design. The use of non-divisor tiling factors (e.g., $T\_{I1}=129$, $T\_{J1}=130$) helps locate a better design that fully utilizes the DSPs, while the divisor-only design only uses 60\% of the DSPs due to the limited choices of tiling factors. This result reveals the necessity of taking non-divisors into consideration to achieve high performance. To accommodate for the non-divisor tiling factors, we introduce the second mutation scheme, \textit{random mutation}. 

\begin{table}[t]
\centering
\caption{Best solutions found by two mutation methods.}
\vspace{-0.05in}
\resizebox{\columnwidth}{!}{
\begin{tabular}{lccccccccc}
\toprule
Mutation Methods         & Throughput      & BRAM & DSP  & T\_{I1} & T\_{J1} & T\_{K1} & T\_{I2} & T\_{J2} & T\_{K2} \\ \toprule
Factorization   & 0.61$\times$ & 48\%    & 60\% & 64   & 128   & 128    & 16     & 4     & 8   \\ \hline
Factorization + Random  & 1.00$\times$ & 99\%    & 100\% & 129   & 130   & 64    & 3     & 13     & 4   \\ \bottomrule
\end{tabular}
}
\label{table:tuning_mutation_cmp}
\vspace{-0.05in}
\end{table}

\textit{2) Random mutation:} We randomly select one loop $l_1$ and mutate this loop by changing the loop bound to a random value $s\in [1,l_1]$. Next, we select another corresponding loop $l_2$ and change its loop bound to a new value $s'$ computed by $s'=ceil(l_1\times l_2/s)$.

We show one example of random mutation in Figure~\ref{fig:evo_mutation}. In the right example, we select the loop $j.2$ and change its original loop size from 8 to 6. Then we choose the other loop $j.1$ to mutate. We compute the new loop bound for $j.1$ as $ceil(4\times8/6)=6$. As a result, the tiling factor $T\_J1$ is changed from 32 to 36, which is a non-divisor of the problem size $J=64$. The random mutation method adjusts two loops at the same time as an effort to keep the product of these two loops with minimal changes. The use of \texttt{ceil()} function guarantees the new product is no less than the original product so that the mutated program is always legal.


When performing the mutation, we assign a probability $\alpha$ to execute the factorization-based mutation and the probability $1-\alpha$ to the random mutation. Based on a grid search, we set $\alpha$ to 0.4 by default.


\paragraph{Crossover}
The crossover operation exchanges the genomes of two individuals. To guarantee the validness of the offspring, we exchange the tiling factors associated with the same original loop together. 


\subsection{Mathematical Programming}
\label{sec:solver_details}
Although the mathematical programming (MP)-based method fails to identify the optimal design, the design it finds achieves relatively good performance and could be used as the initial population of evolutionary search. Odyssey implements a MP-based optimizer to produce high-quality seeds to evolutionary search. We formulate the optimization problem as follows.

\subsubsection{Constraints} A valid hardware design should not overuse the available memory and computation resource. For FPGA designs, we consider the BRAM and DSP usage.
\begin{equation}
\small
U_{mem} \leq Mem_{available},\,\,U_{DSP} \leq DSP_{available}
\vspace{-0.025in}
\end{equation}

\paragraph{Memory resource} We model the memory usage as:
\begin{equation}
\small
U_{mem} = \sum_{m\in Modules}{U_{mem}(m)\times N_m}
\vspace{-0.025in}
\end{equation}
which is a sum of the memory usage of each type of hardware module $m$, computed by multiplying the memory usage of module $m$ ($U_{mem}(m)$) with the total number of such modules $N_m$.
%

\paragraph{Computation resource}
We compute the DSP usage as:
\begin{align}
\small
U_{DSP} &= \sum_{m\in Modules}{U_{DSP}(m)\times N_m} \\
U_{DSP}(m) &= \sum_{op \in Compute\_Ops}{SIMD(op)\times U_{DSP}(op)}
\vspace{-0.025in}
\end{align}

We calculate the total amount of DSPs as a sum of the DSPs consumed by all the hardware modules and the per module DSP usage $U_{DSP}(m)$ is computed as a sum of the products of SIMD factor $SIMD(op)$ of each computation operator $op$ and the number of DSPs consumed by each single SIMD lane $U_{DSP}(op)$. We maintain an internal database for the per operator DSP usage $U_{DSP}(op)$.

%
%
%

\subsubsection{Objective Function}

Given the high complexity of the hardware designs, it is usually hard to have a close-formed performance model suitable for the solvers as the objective function. Instead, previous works chose different high-order functions that impact the performance as the objective functions~\cite{analytical_cnn,marvel,cosa}. 
We conducted an experiment on evaluating the effectiveness of several objective functions.

\paragraph{Objective 1: Computation resource} We use the total DSP usage $U_{DSP}$ as the optimization target. The heuristic is that a design with higher performance should utilize more DSPs.
\begin{equation}
\small
Obj1: min(-U_{DSP})
\vspace{-0.025in}
\end{equation}
\paragraph{Objective 2: Off-chip communication} 
We aggregate the off-chip data movement of all the arrays in the program.

\begin{equation}
\small
Obj2: min\sum_{a\in Arrays} DM(a)
\vspace{-0.025in}
\end{equation}

%

\paragraph{Objective 3: Off-chip communication - computation resource} 
This objective function takes both computation and communication into consideration. Ideally, we would like to maximize the computation resource and reduce the off-chip communication.

\begin{equation}
\small
Obj3: min(\sum_{a\in Arrays} DM(a)-U_{DSP})
\vspace{-0.025in}
\end{equation}

We use the off-the-shelf solver (AMPL~\cite{ampl} with Ipopt~\cite{ipopt}) to implement the optimization problem. All the metrics have been normalized. The best solution obtained from the solver is then fed to the evolutionary search as the initial population. Figure~\ref{fig:tuning_obj} shows the search traces of evolutionary search with different optimization targets. Table~\ref{table:tuning_obj_cmp} compares the performance of designs found by solvers.

\begin{figure}[t]
\includegraphics[width=0.7\columnwidth]{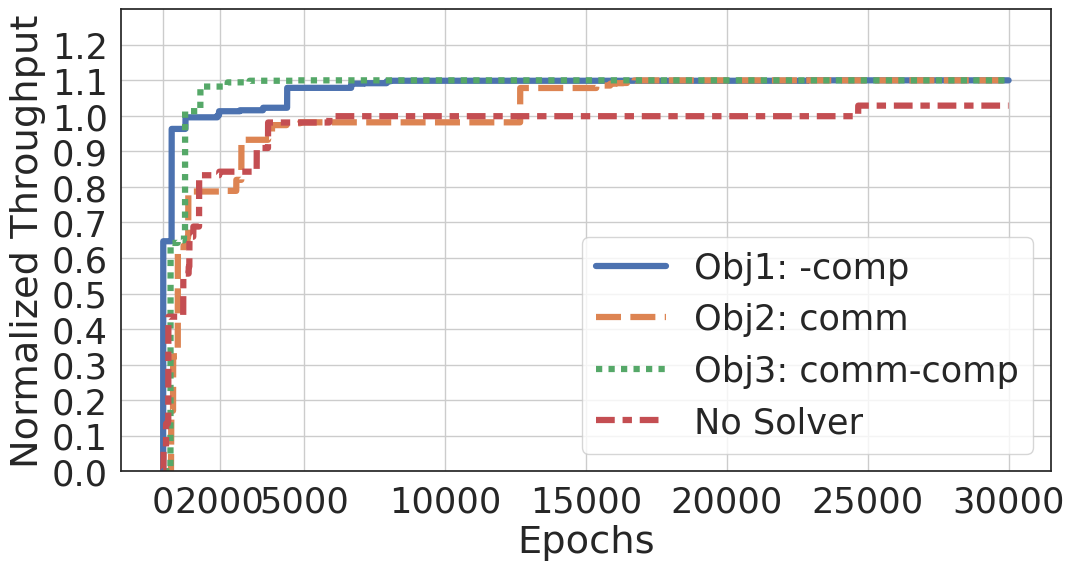}
\centering
\vspace{-0.05in}
\caption{Search traces of evolutionary search initiated with designs found by the MP-based optimizer with different objective functions.}
\vspace{-0.05in}
\label{fig:tuning_obj}
\end{figure}

As seen in the figure, all three objective functions help reduce the convergence time and yield better results compared to the evolutionary search-only method (annotated as \textit{No Solver}). Specifically, Obj3 helps significantly reduce the convergence time. With Obj3, the auto-tuner locates a better design than the baseline (No Solver) within 2000 epochs. In Odyssey, we use the Obj3 as the optimization target of the solver. 
The implications from this experiment are multi-fold. 

First, MP-based methods are insufficient to find optimal designs. As demonstrated in Table~\ref{table:tuning_obj_cmp}, the best design identified by the MP-based approach is 1.5$\times$ slower than the design discovered by Odyssey with a hybrid search method that combines both mathematical programming and evolutionary search.

\begin{table}[t]
\centering
\vspace{-0.05in}
\caption{Comparison of the best solutions found by MP-based methods and Odyssey framework.}
\resizebox{\columnwidth}{!}{
\begin{tabular}{ccccc}
\toprule
Methods & Opt. Target & Latency & Off-Chip Data Movement & DSP   \\ \toprule
Mathematical Programming  & Obj1:-comp       & 1.5$\times$ & 1.7$\times$   & 6.3$\times$ \\ \hline
Mathematical Programming  & Obj2:comm        & 8.6$\times$ & 1.0$\times$   & 1.0$\times$ \\ \hline
Mathematical Programming  & Obj3:comm-comp   & 2.5$\times$ & 1.0$\times$   & 4.0$\times$ \\ \hline
Odyssey    & comm-comp   & 1.0$\times$  & 4.9$\times$       & 6.7$\times$  \\ \bottomrule
\end{tabular}
}
\label{table:tuning_obj_cmp}
\end{table}

Second, the optimal design does not necessarily minimize the off-chip data communication.
In Table~\ref{table:tuning_obj_cmp}, the best design found by Odyssey introduces 4.9$\times$ more off-chip data communication than designs identified by Obj2 and Obj3. 


%% file: tex/results.tex
\section{Evaluation Results}
\label{sec:evaluate_results}
In this section, we evaluate the performance of the Odyssey framework.
We first validate the performance models from Odyssey against the real hardware. Then, we compare the proposed auto-tuning method with several other search methods to assess its efficiency. Lastly, we use Odyssey to investigate the architecture trade-offs of systolic arrays on MM and CNN. 

\subsection{Performance Model Validation}
\subsubsection{Approaches}
\begin{table}[t]
\centering
\caption{Parameters of validation workloads.}
\vspace{-0.025in}
\resizebox{0.7\columnwidth}{!}{
\begin{tabular}{ccc}
\toprule
Workload & Parameters                       & \#Designs \\ \toprule
MM       & [I,J,K]: [64,64,64]              & 18       \\ \hline
CNN      & [I,O,H,W,P,Q]: [16,16,16,16,3,3] & 30       \\ \bottomrule
\end{tabular}
}
\label{table:validation_workload}
\end{table}

\paragraph{Workloads} We evaluate the accuracy of the hardware models derived from Odyssey on two important workloads in HPC and deep learning, MM and CNN. 
Table~\ref{table:validation_workload} details the parameters for the chosen workloads.
We choose a relatively small problem size in consideration of the long RTL simulation time. All the experiments in this paper target Xilinx Alveo U250 FPGA.

\paragraph{Baselines} We measure the design latency by RTL simulation. We use the reported resource numbers from the HLS synthesis reports as the baseline for resource models.

\paragraph{Designs} We validate all 18 and 30 different designs generated for MM and CNN by AutoSA, as listed in Table~\ref{table:tuning_designs}.
To select the tiling factors, we randomly sample the design space.
We compare the performance numbers estimated by the models derived from Odyssey against the numbers measured by RTL simulation and HLS synthesis.

\subsubsection{Validation Results}
\begin{figure}[t]
    \centering
    \subfloat[Latency]{
        \includegraphics[width=0.9\columnwidth]{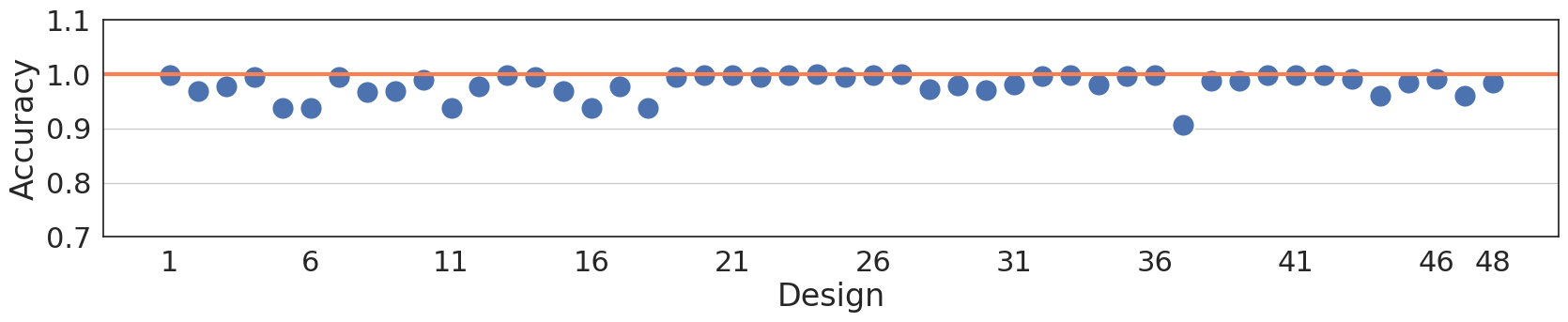}
        \label{fig:tuning_model_time}
    }\\
    \vspace{-0.15in}
    \subfloat[BRAM and DSP]{
        \includegraphics[width=0.9\columnwidth]{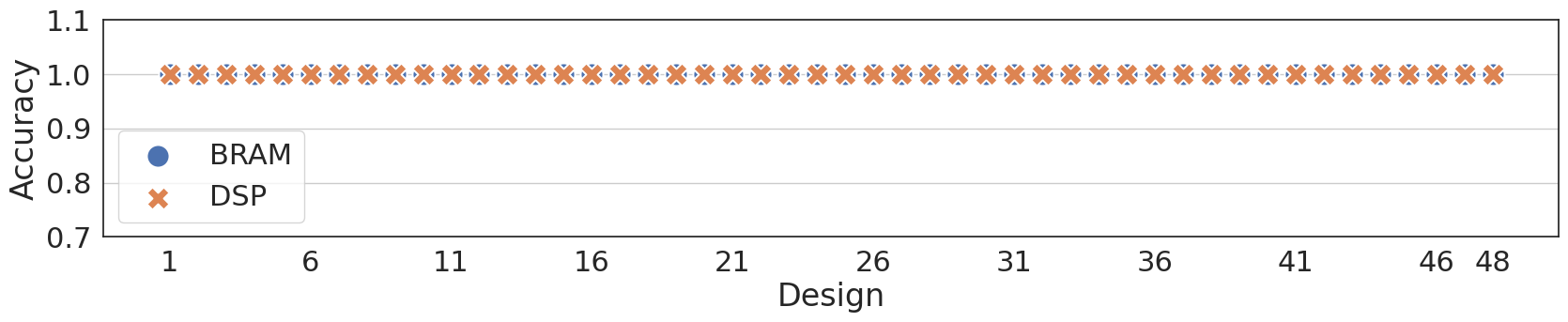}
        \label{fig:tuning_model_bram_dsp}
    }
    \vspace{-0.05in}
    \caption{Performance model validation results.}
    \vspace{-0.025in}
    \label{fig:tuning_model}
\end{figure}

Figure~\ref{fig:tuning_model} compares the model-predicted performance numbers against the measured numbers for all the designs.
The performance models generated by Odyssey are highly accurate, with estimation error rates of 1.99\%, 0\%, and 0\% for latency, BRAM, and DSP, respectively.

\subsection{Auto-Tuning Evaluation}
\label{sec:search_method_cmp}
\subsubsection{Approaches}
\paragraph{Workloads} We evaluate the performance of the proposed auto-tuning method using a $1024\times1024\times1024$ MM. 

\paragraph{Baselines} We use the following methods as baselines.
\begin{itemize}
    \item \textit{Random search}. We randomly sample the design space and update the best solutions. 
    \item \textit{Exhaustive search with pruning}. We extend the random search by pruning the design samples based on the DSP utilization. 
    As the smallest design among the 18 designs uses 30\% of DSPs, we set the DSP pruning threshold to 25\%.
    \item \textit{Simulated annealing}. 
    We use the Python package~\cite{python_annealing} as the baseline. 
    Based on a grid search, we designate the hyper-parameter temperature $T$ to be 200.
    We implement a customized step-taking function using the proposed hybrid mutation method for evolutionary search.
    \item \textit{Bayesian optimization}. 
    We use the Python package~\cite{python_bayesian} as the baseline. 
    \item \textit{OpenTuner~\cite{opentuner}}. OpenTuner is an auto-tuning framework built on an ensemble of several efficient search techniques. It has been demonstrated effective in cases such as searching the GCC compilation flags and optimal schedules for Halide programs~\cite{opentuner}.
    We use the latest release of OpenTuner from its Github repository~\cite{opentuner_github}.
    \item \textit{Reinforcement learning (RL)}. RL is a machine learning algorithm that can be used for hardware optimization.
    The previous work ConfuciuX~\cite{confuciux} implemented a two-step search algorithm for tuning the dataflow architectures which employs RL as the first step to locate a good sub design space and utilizes evolutionary search to perform a more fine-grained search later to find the best design. We use the open-source implementation from ConfuciuX as the RL baseline~\cite{confuciux_github}. ConfuciuX applied a 3-layer multi-layer perceptron (MLP) neural network for the policy network.
\end{itemize}

\paragraph{Designs} We compare our tuning methods against the baselines on all 18 different designs generated for MM by AutoSA. 

\paragraph{Evaluation method} For each systolic array design, we run the search method for 5 minutes and repeat it 3 times. 
The final results are averaged from the 3 runs. All the search methods are executed with a single CPU thread. RL baseline uses Pytorch which implements multi-threading during the training of MLP. All experiments are executed on a workstation with Intel Xeon E5-2680 v4 CPU.

\subsubsection{Results}

\paragraph{Search results quality}
Figure~\ref{fig:tuning_methods_all_cmp} compares the best throughput (1/latency) achieved by each tuning method on the 18 systolic array designs. The throughput is normalized against the optimal performance found by exhaustive search\footnote{We run an exhaustive search until it finishes.}.
Odyssey found design configurations with best performance on 13 designs out of the total 18 designs. 
For the remaining 5 designs, the performance gap is within 1\% of the best performance identified by other baselines (OpenTuner and simulated annealing). 
Overall, Odyssey locates designs that achieve more than 95\% of the optimal performance. 

\begin{figure}[t]
\includegraphics[width=0.9\columnwidth]{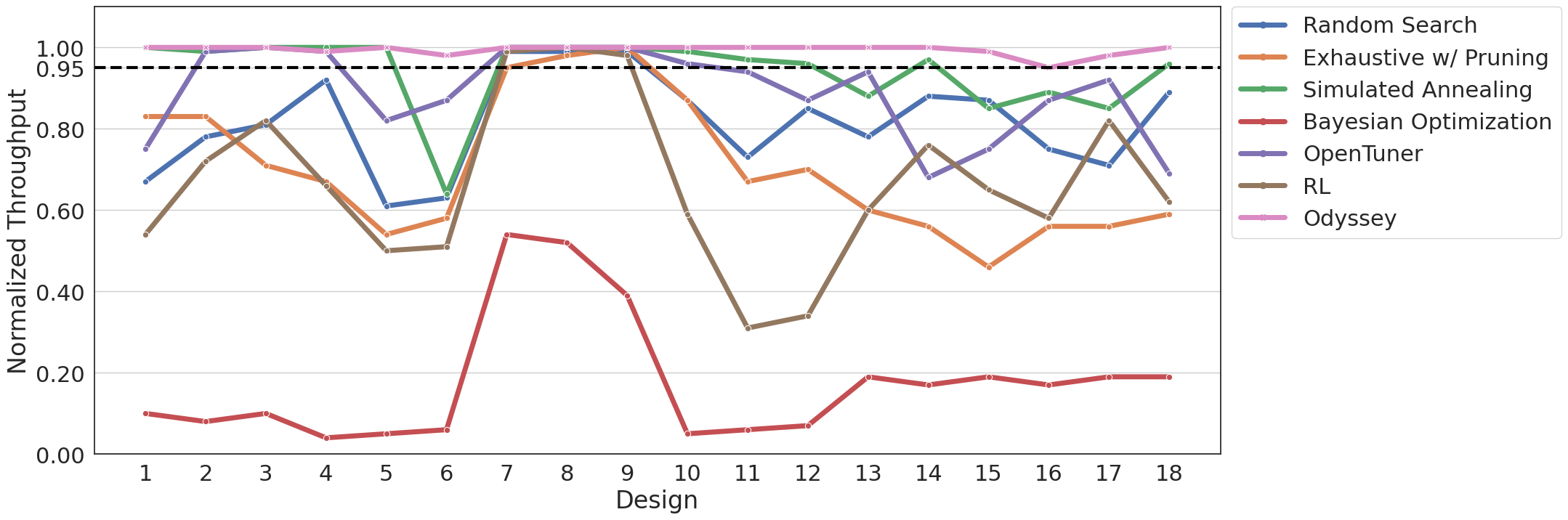}
\centering
\caption{Comparison of the best designs found by different tuning methods.}
\label{fig:tuning_methods_all_cmp}
\end{figure}

\paragraph{Sample efficiency}
In addition to the high-quality search results, Odyssey achieves high sample efficiency. Figure~\ref{fig:tuning_methods_epoch} compares the convergence traces of all the tuning methods on the design with the highest optimal throughput. As shown by the figure, Odyssey finds a good design configuration resulting in 93\% of the optimal performance after evaluating 3000 design samples. Simulated annealing earns the second best performance, locating a design that reaches 66\% of the optimal performance.

\begin{figure}[t]
\includegraphics[width=0.9\columnwidth]{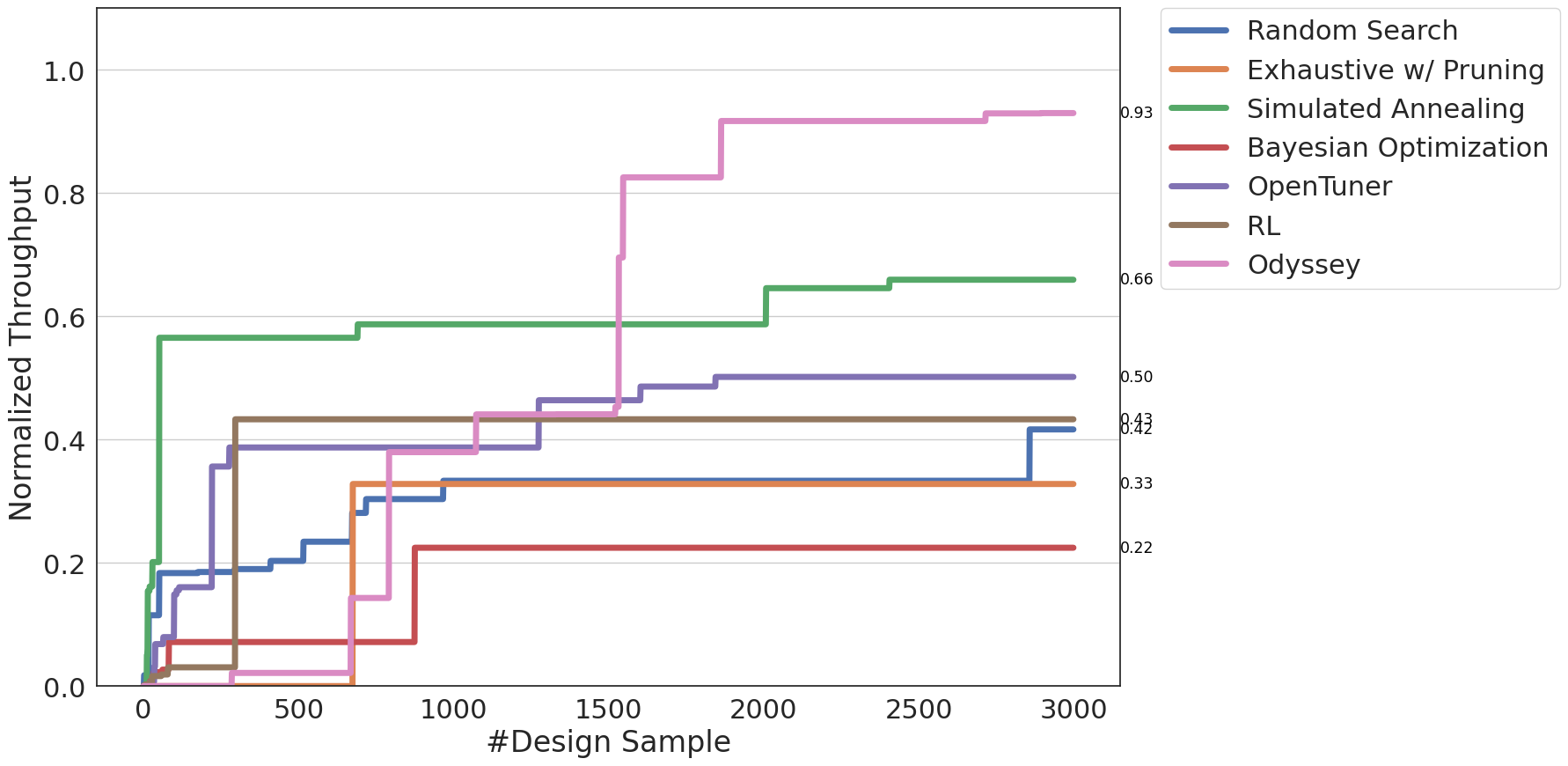}
\centering
\vspace{-0.025in}
\caption{Comparison of sample efficiency of different tuning methods.}
\vspace{-0.025in}
\label{fig:tuning_methods_epoch}
\end{figure}

\paragraph{Search time}
\begin{figure}[t]
\includegraphics[width=0.9\columnwidth]{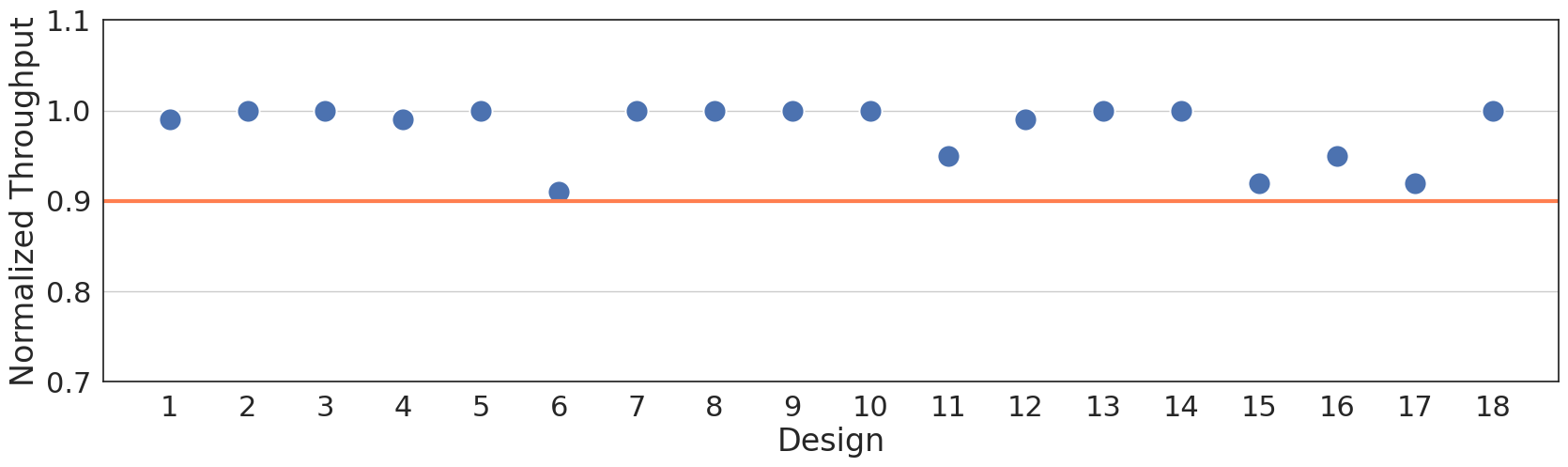}
\centering
\vspace{-0.025in}
\caption{Best designs found in 5 seconds by Odyssey for all different designs of MM.}
\vspace{-0.025in}
\label{fig:tuning_time_limit}
\end{figure}

Lastly, Figure~\ref{fig:tuning_time_limit} shows the performance of the best designs located by Odyssey in 5 seconds. Odyssey finds designs achieving over 90\% of the optimal performance for all the designs in 5 seconds with a single CPU thread. 

\subsection{Application 1: MM}
In this section, we perform a detailed architecture comparison of different systolic arrays for matrix multiplication. Figure~\ref{fig:tuning_mm_cmp} displays the normalized throughput and design usage of different systolic arrays for MM. We draw the following conclusions from this figure.

\begin{figure}[t]
    \centering
    \subfloat[Throughput]{
        \includegraphics[width=0.45\columnwidth]{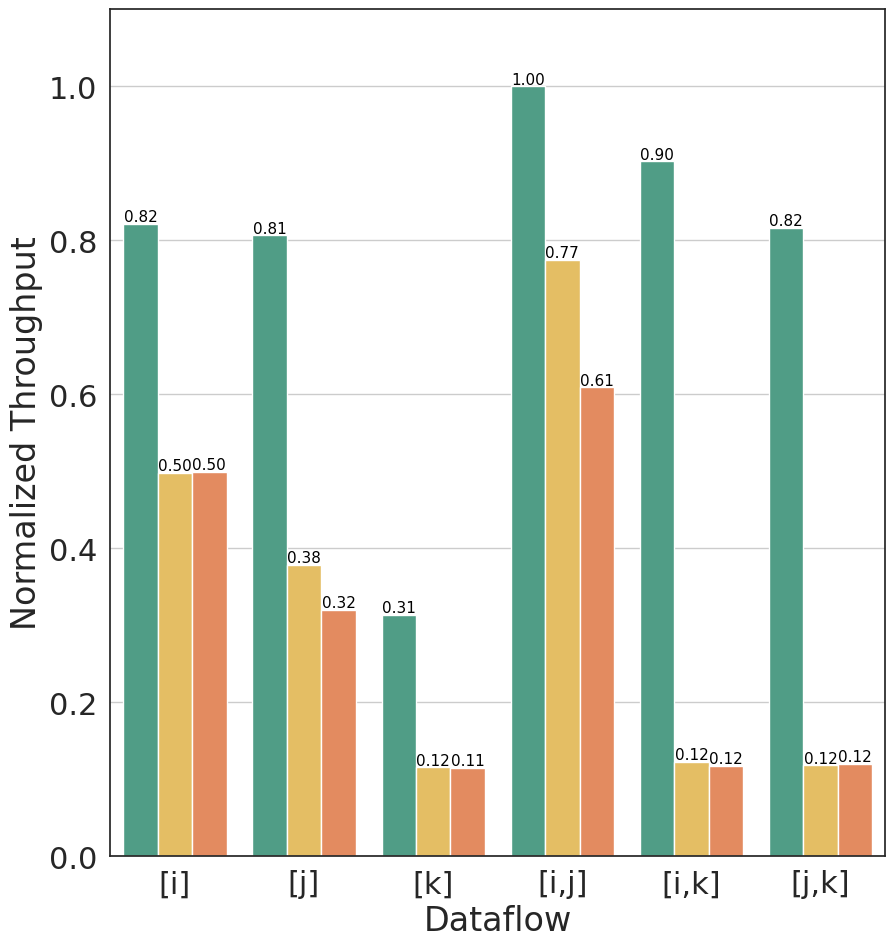}
        \label{fig:tuning_dataflow_mm_latency}
    }
    \subfloat[DSP]{
        \includegraphics[width=0.45\columnwidth]{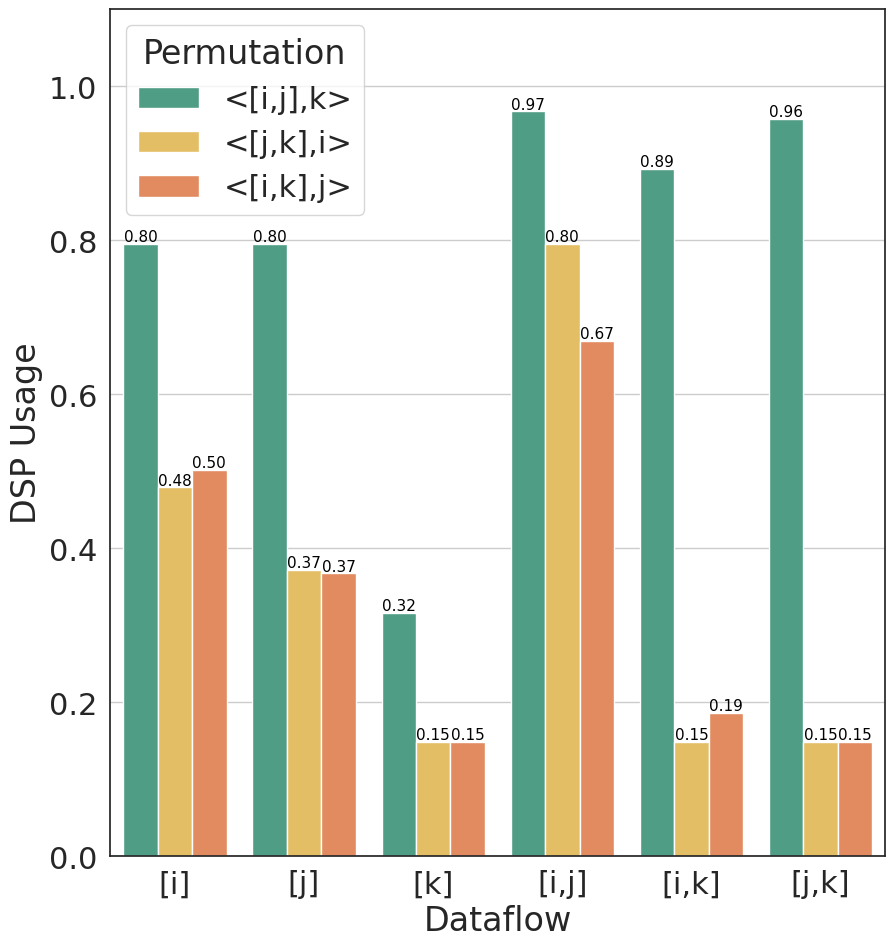}
        \label{fig:tuning_dataflow_mm_latency_dsp}
    }
    \vspace{-0.05in}
    \caption{Normalized Throughput and DSP usage of different designs for 1024x1024x1024 MM.}
    \vspace{-0.05in}
    \label{fig:tuning_mm_cmp}
\end{figure}

\paragraph{Loop permutation} The loop permutation ordering $<[i,j],k>$ dominates the performance across different dataflows. The major reason for such a performance gap is the extra memory consumed by the two other loop orderings. As discussed earlier in Section~\ref{sec:loop_permute}, the loop orderings $<[j,k],i>$ and $<[i,k],j>$ introduce additional I/O modules for loading in the intermediate results, leading to a higher memory usage than the loop ordering $<[i,j],k>$. Table~\ref{table:permute_cmp} illustrates the resource usage breakdown for three designs with the same dataflow $[i]$ and different loop permutation orderings. As shown in the table, designs with loop orderings $<[j,k],i>$ and $<[i,k],j>$ allocate more BRAMs to I/O modules for matrix $C$, which limits the size of array ($\#PE$), leading to a poorer performance than the ordering $<[i,j],k>$.

\begin{table}[t]
\centering
\vspace{-0.05in}
\caption{Resource breakdown of three designs for MM with the same dataflow $[i]$.}
\vspace{-0.025in}
\resizebox{\columnwidth}{!}{
\begin{tabular}{cccccccc}
\toprule
Loop Orderings                       & Latency & \#PE & \multicolumn{5}{c}{BRAM Usage}                                 \\ \toprule
                                     &         &      & Total & I/O Module (A) & I/O Module (B) & I/O Module (C) & PE  \\ \toprule
\textless{}{[}i,j{]},k\textgreater{} & 1.00$\times$ & 1.66$\times$  & 3,621  & 76\%           & 1\%            & 19\%           & 5\% \\ \hline
\textless{}{[}j,k{]},i\textgreater{} & 1.65$\times$ & 1.00$\times$  & 3,519  & 47\%           & 3\%            & 47\%           & 3\% \\ \hline
\textless{}{[}i,k{]},j\textgreater{}& 1.64$\times$ & 1.05$\times$  & 3,594  & 48\%           & 1\%            & 48\%           & 3\% \\ \bottomrule
\end{tabular}
}
\label{table:permute_cmp}
\end{table}

\paragraph{Dataflow} Among the designs with the same loop ordering, the dataflow $[i,j]$ achieves the best performance because the dataflow $[i,j]$ exploits the most degrees of parallelism from the application. For MM, we vectorize the loop $k$ for all the designs which serve as the first degree of parallelism. 
The space dimensions of the systolic array exploit the rest of the degrees of parallelism.
1D systolic arrays can exploit at most one more degree of parallelism. In comparison, the dataflow $[i,j]$ exploits two more degrees of parallelism ($i$ and $j$). Other 2D systolic arrays ($[i,k]$ and $[j,k]$) include $k$ in the space dimensions such that they could extract at most one more degree of parallelism as similar as 1D arrays. More dimensions of parallelism help cover more profitable designs that fully utilize the on-chip resource. As seen in Figure~\ref{fig:tuning_dataflow_mm_latency_dsp}, the dataflow $[i,j]$ utilizes the most DSPs, achieving the highest throughput among all the designs.

\subsection{Application 2: CNN}
\label{sec:cnn_dataflow_cmp}

In the previous section, we perform a detailed performance analysis of different systolic array designs for MM. 
This section applies the same mythology and compares the performance of different systolic arrays for CNN.

As shown in Table~\ref{table:tuning_designs}, Odyssey produces 30 different designs for CNN. We evaluate these designs on the VGG16 network~\cite{vgg}. The VGG16 network includes 13 convoutional (CONV) layers and 3 fully-connected layers. We only examine the mapping of CONV layers in this work. We annotate the 13 CONV layers as CONV[1-13] following their order in the network.
Figure~\ref{fig:tuning_dataflow_conv1_3} presents the detailed evaluation results of the first 2 CONV layers of VGG16. We make the following conclusions from this figure.

\begin{figure}[t]
    \centering
    \subfloat[CONV1]{
        \includegraphics[width=0.45\columnwidth]{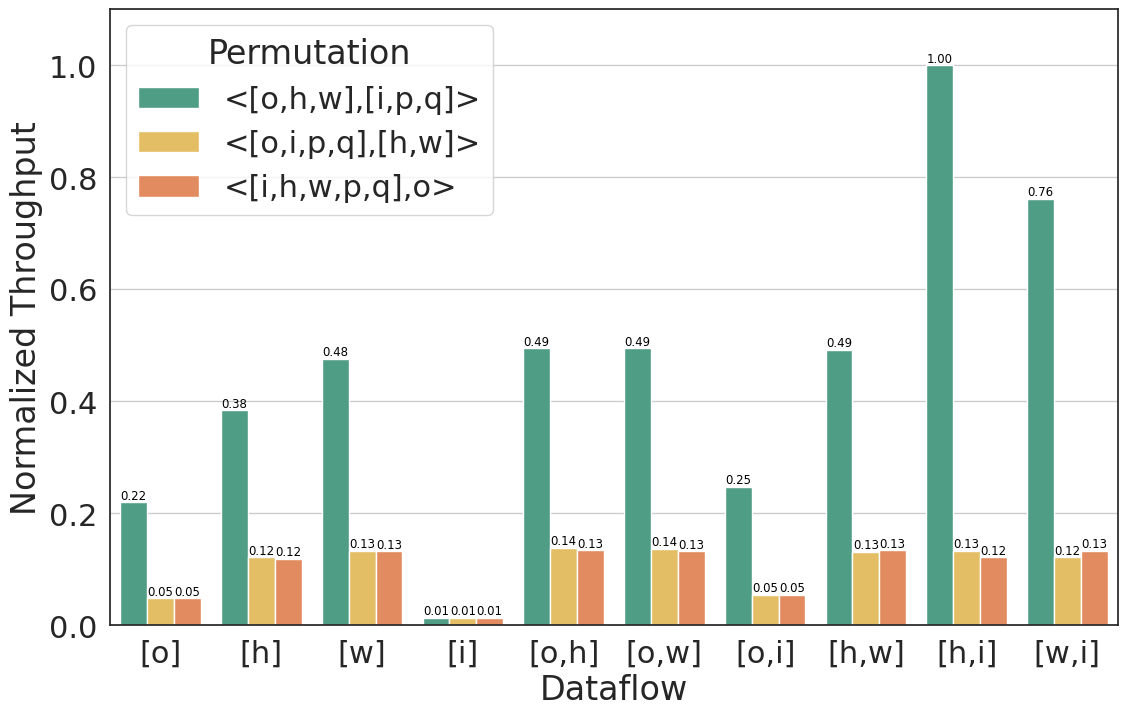}
        \label{fig:tuning_dataflow_conv1}
    }
    \subfloat[CONV2]{
        \includegraphics[width=0.45\columnwidth]{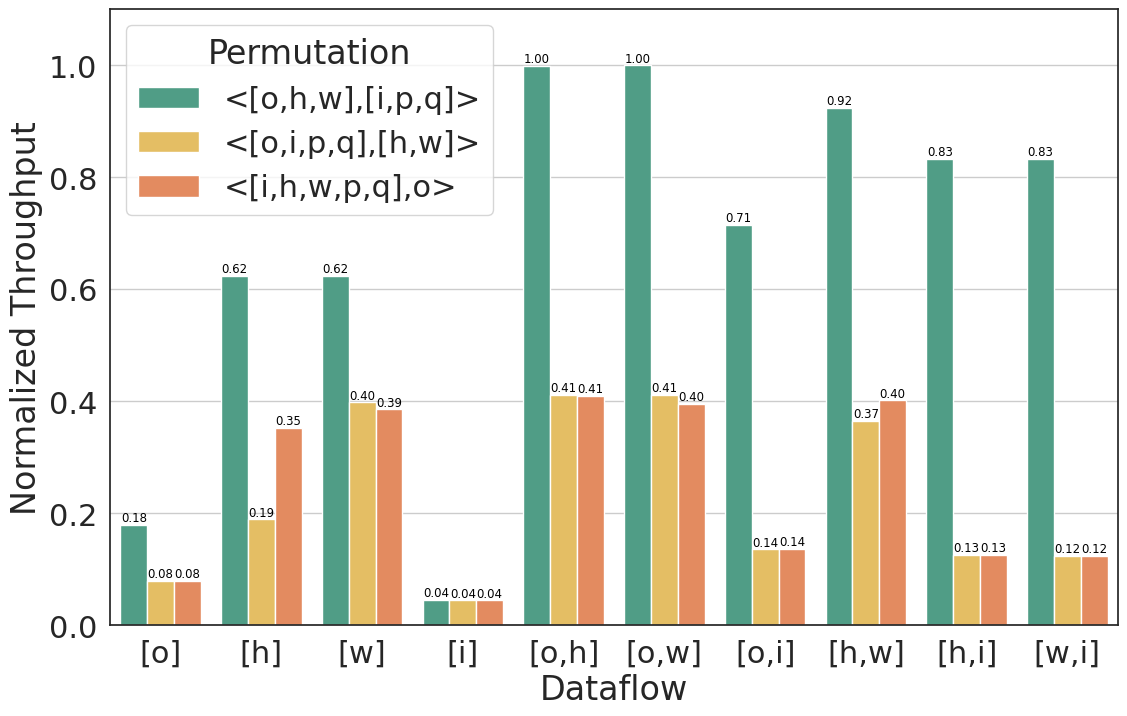}
        \label{fig:tuning_dataflow_conv2}
    }
    \vspace{-0.05in}
    \caption{Normalized throughput of different systolic arrays on CONV1 and CONV2 of VGG16. Throughput of each systolic array design is normalized against the highest number achieved on each CONV layer.}
    \label{fig:tuning_dataflow_conv1_3}
    \vspace{-0.1in}
\end{figure}

\paragraph{Loop permutation} Consistent with the observation from the MM case study, designs with the loop ordering $<[o,h,w],[i,p,q]>$, which eliminates the additional I/O modules for transferring the intermediate results of output feature maps, result in the best performance across all the dataflows.

\paragraph{Dataflow} The performance of different dataflows varies across CONV layers. On CONV1, we observe that dataflow $[h,i]$ delivers the best performance. However, on CONV2, dataflows $[o,h]$ and $[o,w]$ dominate other designs. Below, we examine the best designs on these two CONV layers in detail. 

Figure~\ref{fig:tuning_conv_dataflow_arch_cmp} depicts the detailed architectures of these two designs. AutoSA implements a multi-level I/O network for transferring the data (e.g., L1, L2, L3).
To increase the I/O throughput, we pack multiple data elements between the I/O modules.
A larger I/O bus width introduces more hardware resources. 
As a consequence, the bus width gradually increases at higher levels as a trade-off between the data transfer throughput and hardware resource. 

\begin{figure}[t]
    \centering
    \vspace{-0.05in}
    \subfloat[{[h,i]}]{
        \includegraphics[width=0.43\columnwidth]{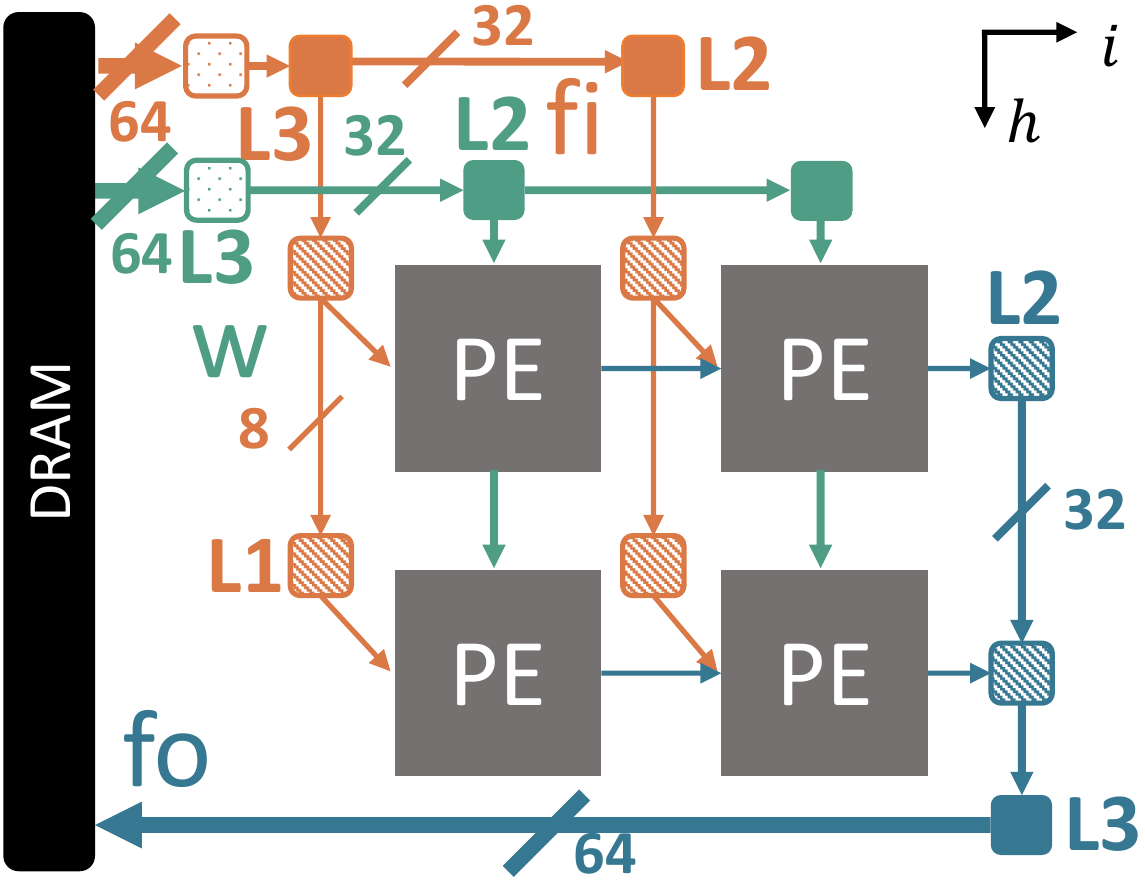}
        \label{fig:tuning_conv_oh}
    }
    \subfloat[{[o,h]}]{
        \includegraphics[width=0.43\columnwidth]{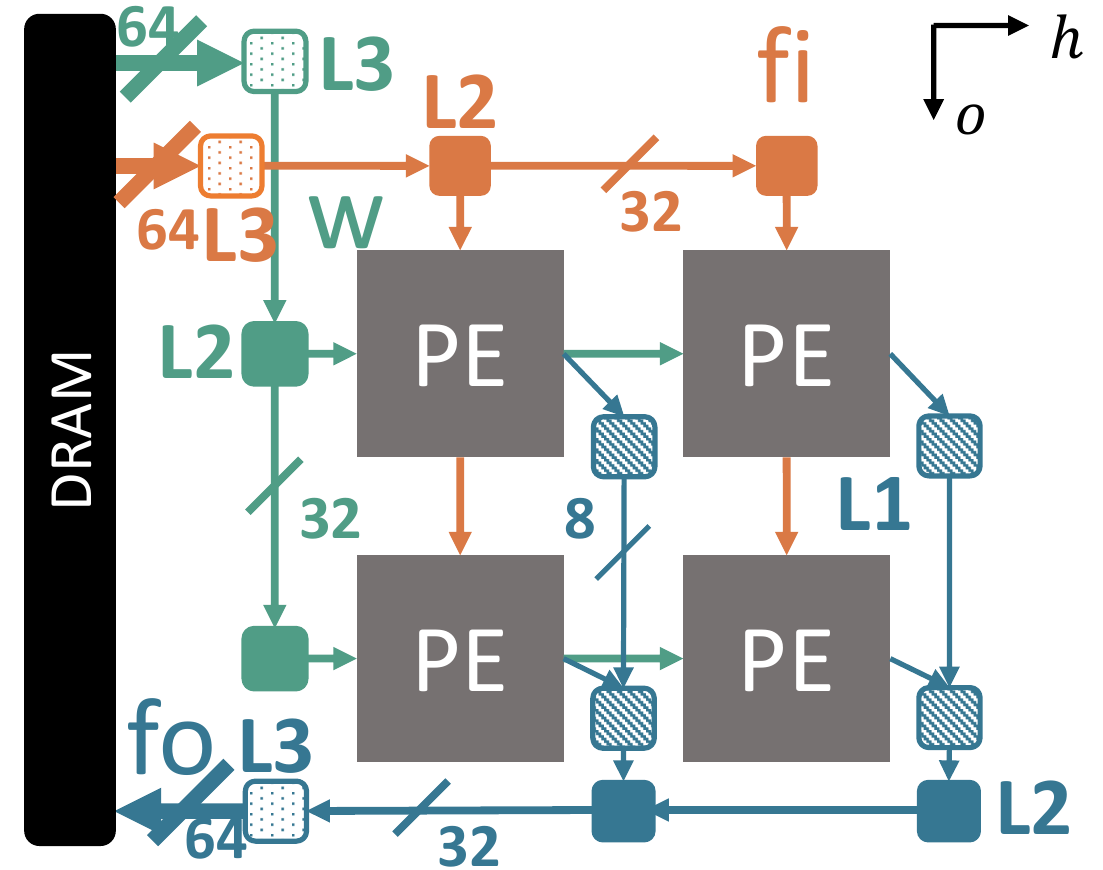}
        \label{fig:tuning_conv_hi}
    }
    \vspace{-0.05in}
    \caption{Architectures of two systolic arrays with dataflows $[h,i]$ and $[o,h]$ for CNN.}
    \vspace{-0.05in}
    \label{fig:tuning_conv_dataflow_arch_cmp}
\end{figure}

\begin{table}[t]
\centering
\caption{Performance comparison of dataflows $[h,i]$ and $[o,h]$ on the CONV1 and CONV2 of VGG16.}
\resizebox{\columnwidth}{!}{
\begin{tabular}{ccccccccccc}
\toprule
CONV Layer         & Problem Sizes  & Dataflow  & SIMD    & Latency & PE Latency & $fo$ I/O Latency & DSP  & DSP Eff. & T\_{I1} & Bottleneck   \\ \toprule
\multirow{2}{*}{1} & \multirow{2}{*}{\begin{tabular}[c]{@{}c@{}}{[}I,O,H,W,P,Q{]}: \\ {[}3,64,224,224,3,3{]}\end{tabular}}  & {[}h.i{]} & {[}i{]} & 2.86$\times$  & 2.85$\times$     & 1.00$\times$(L2)     & 1,120 & 67\%     & \textcolor{burntorange}{4}     & Problem size \\ \cline{3-11} & & {[}o,h{]} & {[}i{]} & 3.99$\times$  & 2.85$\times$     & \textcolor{burntorange}{3.99$\times$(L1)}    & 1,120 & 48\%     & 4     & I/O network  \\ \hline
\multirow{2}{*}{2} & \multirow{2}{*}{\begin{tabular}[c]{@{}c@{}}{[}I,O,H,W,P,Q{]}: \\ {[}64,64,224,224,3,3{]}\end{tabular}} 
& {[}h.i{]} & {[}i{]} & 7.16$\times$ & 7.13$\times$    & 2.14$\times$(L2)     & \textcolor{burntorange}{7,680} &  84\%     & 64    & Parallelism  \\ \cline{3-11} 
& & {[}o,h{]} & {[}i{]} & 5.88$\times$ & 5.85$\times$    & 4.17$\times$(L1)     & 8,320 & 94\%     & 64    &              \\ 
\bottomrule
\end{tabular}
}
\label{table:tuning_conv_analysis}
\end{table}

Table~\ref{table:tuning_conv_analysis} summarizes the performance details of two designs with dataflows $[h,i]$ and $[o,h]$ on first two CONV layers. On CONV1, the design latency of dataflow $[o,h]$ is bound by the data transfer latency of the L1 I/O modules that drain out the output feature maps $fo$. 
In comparison, in dataflow $[h,i]$, data of $fo$ are accumulated across PEs and drained out by L2 I/O modules. 
As more data are packed between the L2 I/O modules than L1 I/O modules (32 Bytes vs. 8 Bytes), the I/O network is no longer the performance bottleneck for dataflow $[h,i]$. Instead, the performance loss comes from the zero padding on the input channel dimension $i$. The best first-level tiling factor $T\_{I1}$ for both dataflows is 4. As a consequence, the input dimension $i$ is padded from 3 to 4 for both dataflows, introducing $25\%$ computation overheads. 

On CONV2, the I/O network is no longer the performance bottleneck for the dataflow $[o,h]$. The design is compute-bound as the input dimension $i$ increases from 3 to 64 on this layer and there is more computation inside the PEs on accumulating the intermediate results along the $i$ dimension. We see from Table~\ref{table:tuning_conv_analysis} that PE latency now dominates the design. On this layer, dataflow $[h,i]$ delivers a slightly lower performance than the $[o,h]$ due to the fewer degrees of parallelism to explore. Both dataflows vectorizes the loop $i$. The dataflow $[o,h]$ explores two more levels of parallelism by mapping the loops $o$ and $h$ to space dimensions. In comparison, dataflow $[h,i]$ only exploits one more level of parallelism on the loop $h$. With more degrees of freedom to use in the design space, we are able to find a design which utilizes more DSPs and achieves a higher performance.

Finally, Figure~\ref{fig:tuning_vgg16_network} presents the best throughput attained by systolic arrays with different dataflows on VGG16. We fix the loop ordering to $<[o,h,w],[i,p,q]>$ for all dataflows. Figure~\ref{fig:tuning_vgg16_network_mean} computes the geometric mean of the best throughput of each dataflow over the entire network. 

\begin{figure}[t]
\includegraphics[width=0.98\columnwidth]{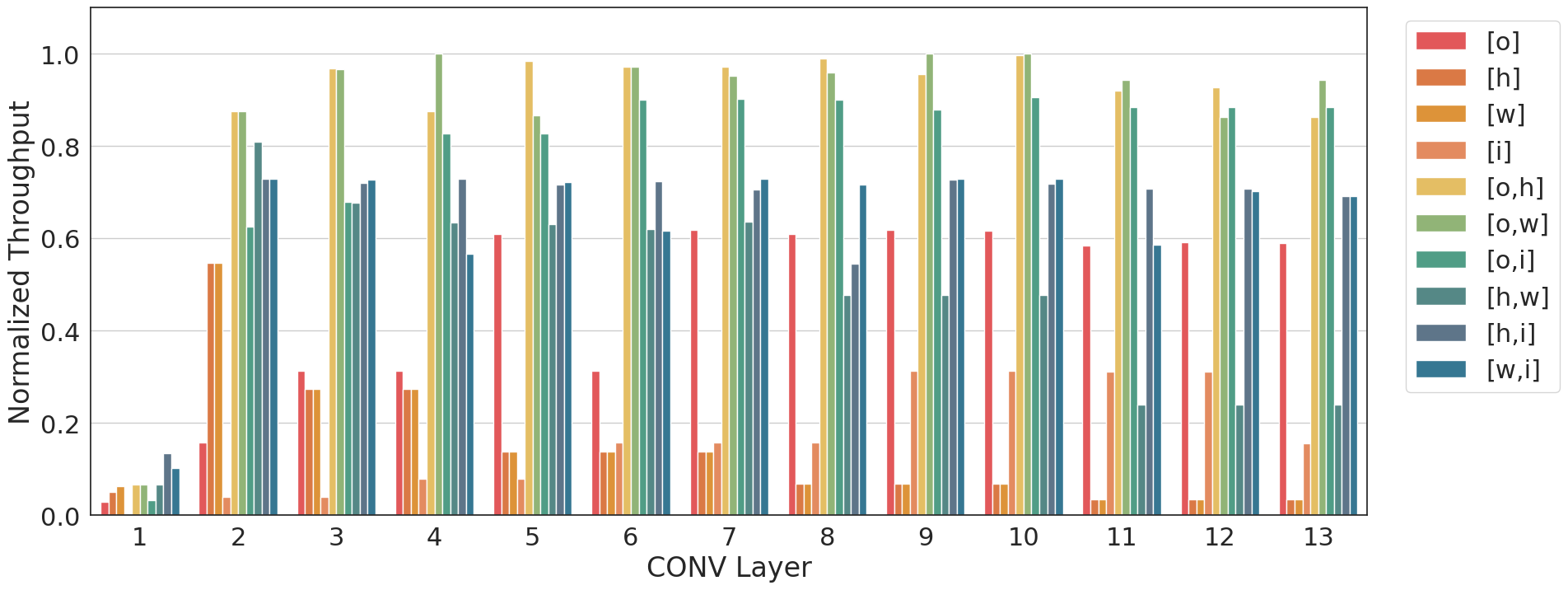}
\centering
\vspace{-0.05in}
\caption{Throughput of different dataflows on VGG16. Throughput of each design is normalized against the highest throughput achieved on the entire network.}
\label{fig:tuning_vgg16_network}
\vspace{-0.05in}
\end{figure}

\begin{figure}[t]
    \centering
    \subfloat[VGG16]{
        \includegraphics[width=0.43\columnwidth]{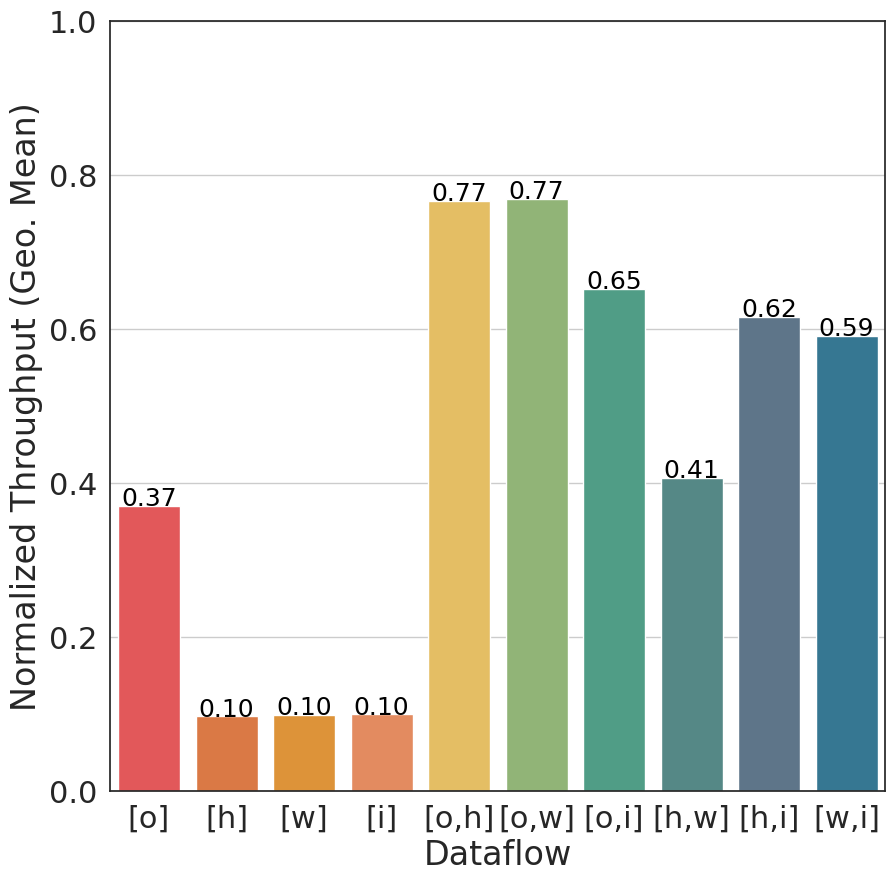}
        \label{fig:tuning_vgg16_network_mean}
    }
    \subfloat[ResNet50]{
        \includegraphics[width=0.43\columnwidth]{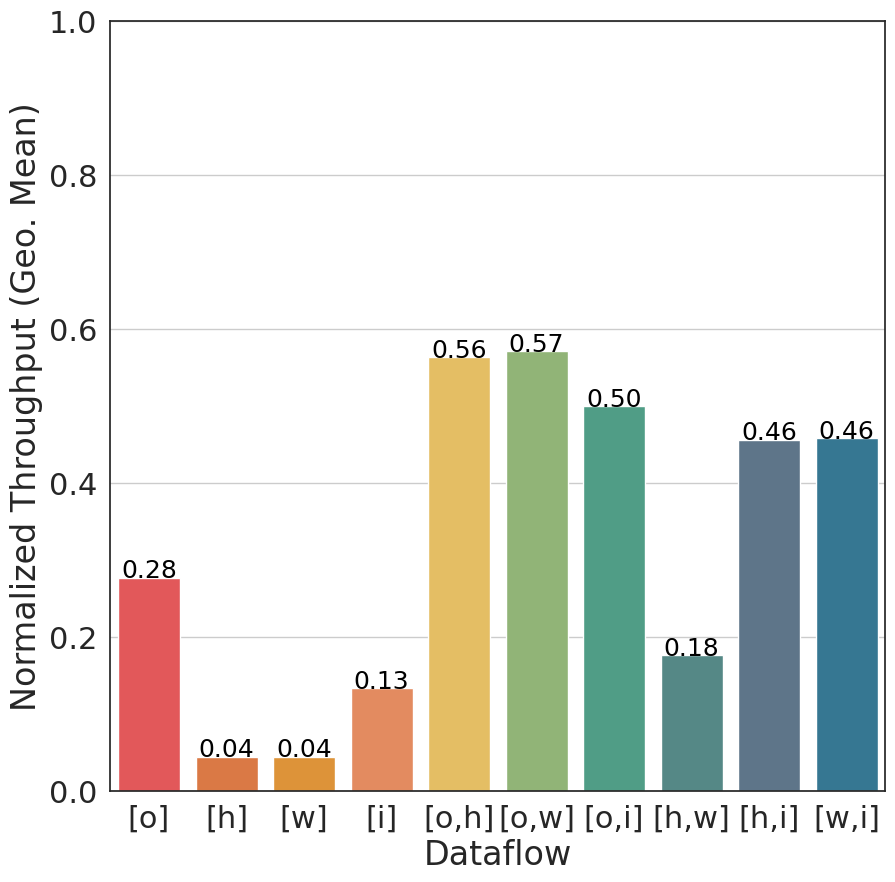}
        \label{fig:tuning_resnet50_network_mean}
    }
    \vspace{-0.05in}
    \caption{Geo. Mean of the throughput achieved by different dataflows across all the CONV layers in VGG16 and ResNet50.}
    \vspace{-0.1in}
    \label{fig:tuning_network_mean}
\end{figure}

Overall, we observe that 2D systolic arrays achieve a higher performance than 1D arrays with more degrees of parallelism to exploit. Dataflows $[o,h]$ and $[o,w]$ deliver the best performance among all the dataflows, followed by $[o,i]$, $[h,i]$, and $[w,i]$. The best average throughput of a single dataflow on VGG16 is 77\% of the peak layer-wise throughput obtained on the entire network. This is due to the divergent problem sizes of each CONV layer that limit the performance of a single systolic array. We applied the same methodology on another CNN, ResNet50~\cite{resnet}, and summarized the results in Figure~\ref{fig:tuning_resnet50_network_mean}. ResNet50 is more compact than VGG16, with less parallelism and data locality to exploit. 
This results in a reduction of the best average throughput of a single array to 57\% of the peak layer-wise throughput on ResNet50.

Ideally, we would like to achieve 100\% average throughput such that computation resource is fully utilized for each CONV layer. However, in reality, we ended with only 77\% and 57\% for VGG16 and ResNet50, respectively. This indicates that mapping CNNs to a single systolic array will face the issue of resource underutilization. 
Several previous works also identified this issue~\cite{tgpa,dnnexplorer,herald,shenmaximizing,jouppi2020domain}.
In terms of solutions, TPU chose to increase the batch size to improve the overall throughput~\cite{jouppi2020domain}. On FPGAs, previous works like DNNExplorer~\cite{dnnexplorer} implemented a multi-array architecture that customizes each smaller array to different CONV layers to improve the overall throughput. 



%% file: tex/discuss.tex
\section{Lessons Learned from This Work}
\label{sec:discuss}
This section summarizes four major lessons we learned from this work.

\paragraph{Lesson 1: Incomplete and inaccurate design space modeling could lead to inferior search results and skewed architecture conclusions.} We have demonstrated that considering only divisor tiling factors or overlooking the prologue and epilogue phases could lead to inferior designs. This could further hurt the validity of the architecture conclusions derived from these results. For example, Figure~\ref{fig:tuning_mm_cmp_skew} illustrates the search results for MM with only divisor tiling factors. With the limited design space, the best systolic array design only uses 60\% of the DSP. Furthermore, architects may draw skewed conclusions from these results such as all three 2D array dataflows achieve the equivalent performance, which contradicts to the findings as seen in Figure~\ref{fig:tuning_mm_cmp}.

\begin{figure}[t]
    \centering
    \vspace{-0.05in}
    \subfloat[Throughput]{
        \includegraphics[width=0.47\columnwidth]{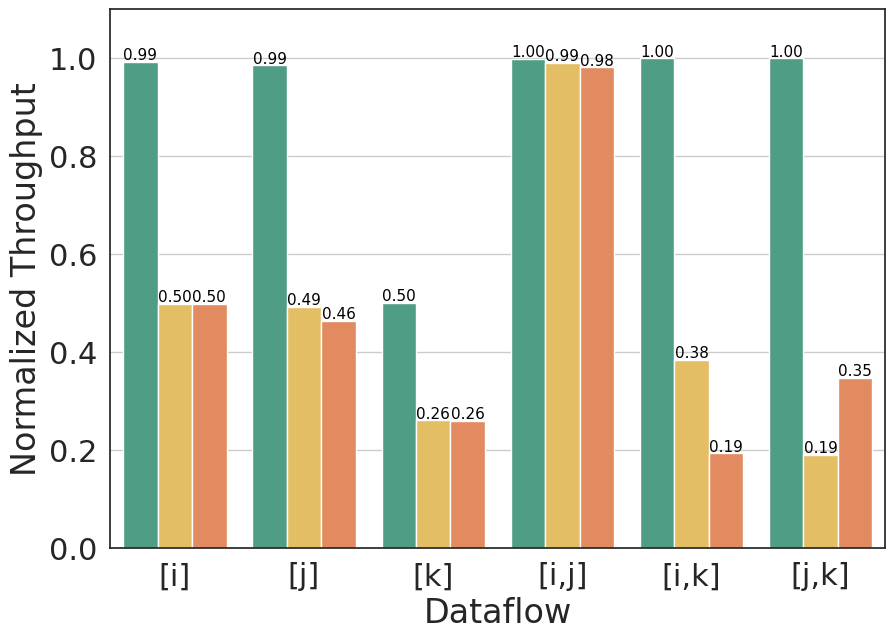}
        \label{fig:tuning_dataflow_mm_latency_skew}
    }
    \subfloat[DSP]{
        \includegraphics[width=0.47\columnwidth]{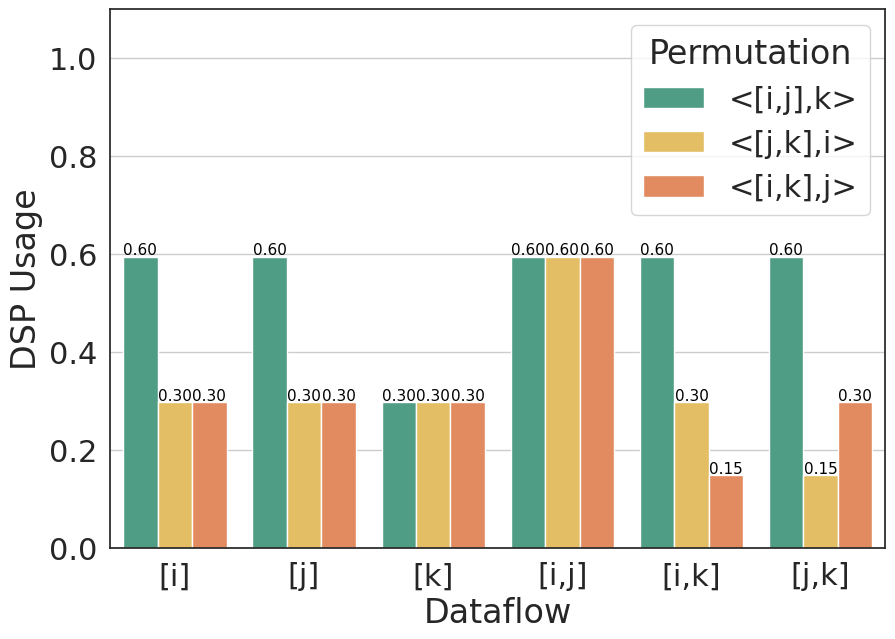}
        \label{fig:tuning_dataflow_mm_latency_dsp_skew}
    }
    \vspace{-0.05in}
    \caption{Throughput and DSP usage of different designs with only divisor tiling factors for 1024x1024x1024 MM.}
    \vspace{-0.05in}
    \label{fig:tuning_mm_cmp_skew}
\end{figure}

\paragraph{Lesson 2: A high-performance design needs to balance computation and communication.} We show that the latency-optimal design does not necessarily minimize the off-chip data communication and both computation and communication need to be considered. Several previous works~\cite{chen_communication,marvel} that used the off-chip communication to prune the design space may leave out the optimal designs, leading to inferior performance.

\paragraph{Lesson 3: Iterative search methods can be enhanced with mathematical programming solvers.} 
Iterative methods achieve high sample efficiency and are portable to different search problems. 
It can be further enhanced with mathematical programming solvers.
We have demonstrated that a MP-based optimizer which uses a simplified objective function that considers data communication and computation resource can effectively boost the convergence of evolutionary search. 
Although this work only focuses on systolic arrays, the methodology is general, and designers can apply it to performance tuning tasks for other applications and hardware architectures.

\paragraph{Lesson 4: Compiler-assisted design space construction improves productivity and performance.} To the best of our knowledge, Odyssey is the first work that constructs the design space of systolic arrays with an open-source compiler that generates real hardware designs. This approach guarantees the validity and accuracy of search results and improves the productivity. In the recent years, we have seen an emergence of domain-specific language compilers for DSAs~\cite{spatial,genesis,gemmini,fleet}. We believe the methodology in this work could provide valuable insights to further enhance the performance tuning of these works and enable a wider adoption of DSAs in the post Moore's law era.






%% file: tex/conclude.tex
\section{Limitations}
\label{sec:limit}
The current framework limits the input programs to ones with rectangular iteration domains. For programs with non-rectangular iteration domains, such as LU decomposition, programmers need to implement the design performance models in Python manually and supply it to the auto-tuner. We will address this limitation in the future work.

\section{Conclusion}
\label{sec:conclude}

This paper presents Odyssey, an automatic design space exploration framework for systolic arrays. Odyssey covers a comprehensive and accurate design space of systolic arrays and incorporates a hybrid search method consisting of the MP-based optimizer and evolutionary search. Furthermore, it implements an effective loop permutation pruning algorithm that helps reduce the design space without leaving out the optimal designs. 
With Odyssey, we unveil the limitations of several commonly used assumptions by previous works in the performance optimization of systolic arrays. 
We have also conducted an architecture analysis of systolic arrays on two applications, MM and CNN. Although this work focuses on the systolic array architecture, the methodology and insights can be applied to other hardware and architectures as well. Future work includes extending the Odyssey framework to support applications with non-rectangular iteration domains.

